\RequirePackage{tikz}
\RequirePackage{tikz-cd}
\documentclass[default,pdflatex]{sn-jnl}

\usepackage{enumitem}

\theoremstyle{thmstyleone}%
\newtheorem{theorem}{Theorem}[section]%
\newtheorem{proposition}[theorem]{Proposition}%
\newtheorem{lemma}[theorem]{Lemma}%
\newtheorem{corollary}[theorem]{Corollary}%
\theoremstyle{thmstyletwo}%
\newtheorem{example}[theorem]{Example}%
\theoremstyle{thmstylethree}%
\newtheorem{definition}[theorem]{Definition}

\newcommand{\ie}{\textit{i.e.}, }
\newcommand{\eg}{\textit{e.g.}, }

\newcommand*{\wrt}{{w.r.t.}}

\newcommand{\cA}{\mathcal{A}}
\newcommand{\cR}{\mathcal{R}}

\newcommand{\ra}{\mathop{\rightarrow}}
\newcommand{\rra}{\mathop{\Rightarrow}}
\newcommand{\lpa}{\mathop{\leadsto}}
\newcommand{\hra}{\mathop{\hookrightarrow}}

\newcommand{\nat}{\mathbb{N}}
\newcommand{\termset}{T(\Sigma, X)}
\newcommand{\termsett}{T(\Sigma \cup \{\square\}, X)}
\newcommand{\subsset}{S(\Sigma, X)}
\newcommand{\length}[1]{\vert {#1} \vert}

\newcommand{\vars}{\operatorname{\mathit{Var}}}
\newcommand{\dom}{\operatorname{\mathit{Dom}}}
\newcommand{\mgu}{\operatorname{\mathit{mgu}}}
\newcommand{\pos}{\operatorname{\mathit{Pos}}}
\newcommand{\subterm}[2]{{#1} \vert_{#2}}

\newcommand{\symb}[1]{\mathsf{#1}}
\newcommand{\asym}{\operatorname{\mathsf{a}}}
\newcommand{\bsym}{\operatorname{\mathsf{b}}}
\newcommand{\csym}{\operatorname{\mathsf{c}}}
\newcommand{\dsym}{\operatorname{\mathsf{d}}}
\newcommand{\fsym}{\operatorname{\mathsf{f}}}
\newcommand{\Fsym}{\operatorname{\mathsf{F}}}
\newcommand{\gsym}{\operatorname{\mathsf{g}}}
\newcommand{\Gsym}{\operatorname{\mathsf{G}}}
\newcommand{\hsym}{\operatorname{\mathsf{h}}}
\newcommand{\Hsym}{\operatorname{\mathsf{H}}}
\newcommand{\ssym}{\operatorname{\mathsf{s}}}
\newcommand{\zero}{\mathsf{0}}
\newcommand{\one}{\mathsf{1}}
\newcommand{\psym}{\operatorname{\mathsf{p}}}
\newcommand{\qsym}{\operatorname{\mathsf{q}}}

\newcommand{\ars}{(A,\rra_{\Pi})}
\newcommand{\oc}{\operatorname{\mathit{OC}}}
\newcommand{\unf}{\operatorname{\mathit{Unf}}}
\newcommand{\binunf}{\operatorname{\mathit{Binunf}}}

\newcommand{\phiinst}{\operatorname{\mathit{ins}}}
\newcommand{\phimg}{\operatorname{\mathit{mg}}}

\newcommand{\sequence}[1]{\langle{#1}\rangle}
\newcommand{\seqset}[1]{{\overline{{#1}}}}

\newcommand{\aprove}{\textsf{AProVE}}
\newcommand{\autonon}{\textsf{AutoNon}}
\newcommand{\mnm}{\textsf{MnM}}
\newcommand{\muterm}{\textsf{MU-TERM}}
\newcommand{\natt}{\textsf{NaTT}}
\newcommand{\nti}{\textsf{NTI}}

\newcommand{\tttt}{%
 \textsf{T\kern-0.15em\raisebox{-0.55ex}T\kern-0.15emT\kern-0.15em\raisebox{-0.55ex}2}%
}

\usepackage{svg}
\newcommand{\orcid}[1]{\href{https://orcid.org/#1}{\includesvg[height = 2ex]{svg-inkscape/ORCID_iD}}}

\begin{document}

\title{Non-Termination in Term Rewriting and Logic Programming}

\author{\fnm{\'Etienne} \sur{Payet} \orcid{0000-0002-3519-025X}}
\email{etienne.payet@univ-reunion.fr}
\affil{
  \orgdiv{LIM},
  \orgname{Universit\'e de la R\'eunion},
  \orgaddress{\country{France}}}

\abstract{In this paper, we define two particular forms of
non-termination, namely \emph{loops} and \emph{binary chains},
in an abstract framework that encompasses term rewriting and
logic programming. The definition of loops relies on the
notion of \emph{compatibility} of binary relations. We also
present a syntactic criterion for the detection of a special
case of binary chains. Moreover, we describe our implementation
\nti{} and compare its results at the Termination Competition
2023 with those of leading analyzers.
}


\maketitle

\section{Introduction}\label{sect:intro}
This paper is concerned with the abstract treatment of
non-termination in structures where one rewrites elements
using indexed binary relations. Such structures can be
formalised by \emph{abstract reduction systems}
(ARSs)~\cite{baaderN98,terese03}, \ie pairs $\ars$ where
$A$ is a set and $\rra_{\Pi}$ (the rewrite relation) is
a union of binary relations on $A$, indexed by a set $\Pi$,
\ie $\rra_{\Pi} = \bigcup \{\rra_{\pi} \mid \pi \in \Pi\}$.
Non-termination in this context corresponds to the
existence of an infinite rewrite sequence 
$a_0 \rra_{\pi_0} a_1 \rra_{\pi_1} \cdots$.
In this introduction, we provide an extended, intuitive,
description of the context and content of the paper, with
several examples. 
In Section~\ref{sect:intro-trs-lp},
we introduce two concrete instances of ARSs. Then, in
Sections~\ref{sect:intro-loops}
and~\ref{sect:intro-binary-chains},
we consider two special forms of non-termination. 
Finally, in Section~\ref{sect:intro-simplified-setting},
we describe our implementation.

\subsection{Term Rewriting and Logic Programming}
\label{sect:intro-trs-lp}
\emph{Term rewriting} and \emph{logic programming} are
concrete instances of ARSs, where $\Pi$ indicates what
\emph{rule} is applied at what \emph{position}.
A crucial difference is that the rewrite relation of term
rewriting (denoted by $\ra_{\Pi}$) relies on instantiation
while that of logic programming (denoted by $\lpa_{\Pi}$)
relies on narrowing, \ie on unification.

\begin{example}\label{ex:trs-intro}
    In term rewriting, rules are pairs of finite terms.
    In this context, a \emph{term rewrite system (TRS)}
    is a set of rules over a given signature (\ie a set
    of function symbols).
    Consider for instance the TRS that consists of the rules:
    \begin{align*}
        r_1 &= \big(\fsym(x), \gsym(\hsym(x,\one),x)\big)  \\
        r_2 &= \big(\one, \zero\big) \\
        r_3 &= \big(\hsym(x,\zero), \fsym^2(x)\big) = (u_3, v_3)
    \end{align*}
    Here, $\gsym,\hsym$ (resp. $\fsym$) are function symbols of
    arity 2 (resp. 1), $\zero,\one$ are constant symbols (\ie
    function symbols of arity 0), $x$ is a variable and $\fsym^2$
    denotes 2 successive applications of $\fsym$.
    Consider the term $s = \gsym\big(\hsym(\fsym(x),\zero),x\big)$.
    Its subterm $s' = \hsym(\fsym(x),\zero)$ results from applying
    the substitution $\theta = \{x \mapsto \fsym(x)\}$ to $u_3$,
    denoted as $s' = u_3\theta$, hence it is an instance of $u_3$.
    Therefore, we have $s \ra_{(r_3,p)} \gsym(v_3\theta,x)$, where
    $p$ is the position of $s'$ in $s$ and $\gsym(v_3\theta,x)$
    results from replacing $s'$ by $v_3\theta$ in $s$, \ie
    $s \ra_{(r_3,p)} \gsym(\fsym^3(x),x)$.
\end{example}

\begin{example}\label{ex:lp-intro}
    In logic programming, one rewrites goals
    (\ie finite sequences of terms) into goals and a rule
    is a pair $(u,\seqset{v})$ where $u$ is a term and
    $\seqset{v}$ is a goal. Moreover, the rewriting always
    takes place at the root position of an element of a
    goal. In this context, a \emph{logic program (LP)}
    is a set of rules over a given signature.
    Consider for instance the LP that consists of the rule:
    \[r = \big(\underbrace{\psym(\fsym(x,\zero))}_u,
    \underbrace{\sequence{\psym(x),\qsym(x)}}_{\seqset{v}}\big)\]
    Here, $\psym,\qsym$ (resp. $\fsym$) are function symbols 
    of arity~1 (resp. 2), $\zero$ is a constant symbol and
    $x$ is a variable. Consider the goal
    $\seqset{s} = \sequence{\psym(\zero),\psym(\fsym(x,x)),\psym(x)}$.
    The rule $(u_1,\sequence{v_1,v'_1}) =
    \big(\psym(\fsym(x_1,\zero)), \sequence{\psym(x_1),\qsym(x_1)}\big)$
    is obtained by renaming the variables of $r$ with variables not
    occurring in $\seqset{s}$. The term $s_1 = \psym(\fsym(x,x))$
    is an element of $\seqset{s}$ and the substitution
    $\theta = \{x\mapsto\zero, x_1\mapsto\zero\}$ is a unifier
    of $s_1$ and $u_1$, \ie $s_1\theta = u_1\theta$ (to be precise,
    $\theta$ is \emph{the most general} unifier of $s_1$ and $u_1$,
    see Section~\ref{sect:subs}).
    Therefore, we have
    $\seqset{s} \lpa_{(r,p)}
    \sequence{\psym(\zero),v_1,v'_1,\psym(x)}\theta$,
    where $p$ is the position of $s_1$ in $\seqset{s}$ and
    $\sequence{\psym(\zero),v_1,v'_1,\psym(x)}$ is the goal
    obtained from $\seqset{s}$ by replacing $s_1$ by the
    elements of $\sequence{v_1,v'_1}$, \ie
    $\seqset{s} \lpa_{(r,p)}
    \sequence{\psym(\zero),\psym(\zero),\qsym(\zero),\psym(\zero)}$.
    The operation which is used to rewrite $\seqset{s}$ into
    $\sequence{\psym(\zero),v_1,v'_1,\psym(x)}\theta$ is called
    \emph{narrowing}.
    In this example, it is applied at a given position
    of a goal using an LP rule, but it can be applied
    similarly at a position $p$ of a term $s$ using a TRS rule
    (\ie renaming of the rule with new variables,
    unification of the subterm of $s$ at position $p$ with
    the left-hand side of the renamed rule, replacement of the
    subterm by the right-hand side of the renamed rule and
    application of the most general unifier to the whole
    resulting term).
\end{example}

Another difference between term rewriting and logic programming
is the effect of a rule application outside the position of
the application. For instance, in Example~\ref{ex:trs-intro}, 
the application of $r_3$ at position $p$ of $s$ only affects
the subterm of $s$ at position $p$. In contrast, in
Example~\ref{ex:lp-intro}, the application of $r$ at position
$p$ of $\seqset{s}$ also affects the terms at the other positions,
because first $s_1$ is replaced by $v_1$ and $v'_1$ and then
the unifier $\theta$ is applied to the whole resulting goal.

Due to these differences, termination in logic programming
does not seem related to termination in term rewriting, but
rather to termination in the construction of the
\emph{overlap closure}~\cite{guttagKM83}, because both use
narrowing.
The overlap closure of a TRS $\cR$, denoted as $\oc(\cR)$, is
the (possibly infinite) set defined inductively as follows.
It exactly consists of the rules of $\cR$, together with the
rules resulting from narrowing elements of $\oc(\cR)$ with
elements of $\oc(\cR)$.
More precisely, suppose that $(u_1,v_1)$ and $(u_2,v_2)$
are in $\oc(\cR)$. Then,
\begin{itemize}
    \item (Forward closure) $(u_1\theta,v'_1) \in \oc(\cR)$
    where $\big(v_1 \lpa_{((u_2,v_2),p)} v'_1\big)$, $\theta$
    is the unifier used and $p$ is the position of a
    non-variable subterm of $v_1$;
    \item (Backward closure) $(u'_2,v_2\theta) \in \oc(\cR)$
    where $\big(u_2 \lpa_{((v_1,u_1),p)} u'_2\big)$, $\theta$
    is the unifier used and $p$ is the position of a
    non-variable subterm of $u_2$.
\end{itemize}
Consider for instance the TRS $\cR$ that consists of the
rule $r = \big(\fsym(\ssym(x),y), \fsym(x,\ssym(y))\big)$.
It admits no infinite rewrite sequence, but the construction
of its overlap closure will not stop as it produces each
$\big(\fsym(\ssym^k(x),y), \fsym(x,\ssym^k(y))\big)$, $k > 0$
(where $\ssym^k$ denotes $k$ successive applications of the
unary function symbol $\ssym$). Now, consider the LP that
consists of $r$ (with delimiters $\sequence{\cdot}$
around the right-hand side). It admits the
infinite rewrite
\[\sequence{\,\underbrace{\fsym(x,y)}_{a_0}\,} \lpa^{\theta_1}_{\pi}
\sequence{\,\underbrace{\fsym(x_1,\ssym(y))}_{a_1}\,} \lpa^{\theta_2}_{\pi}
\sequence{\,\underbrace{\fsym(x_2,\ssym^2(y))}_{a_2}\,} \lpa^{\theta_3}_{\pi}
\sequence{\,\underbrace{\fsym(x_3,\ssym^3(y))}_{a_3}\,} \lpa^{\theta_4}_{\pi} \cdots\]
where $\pi = (r,\sequence{1})$ ($\sequence{1}$ is the position of the
leftmost term of a goal) and each step is decorated
with the corresponding unifier:
\begin{align*}
    \theta_1 &= \{x \mapsto \ssym(x_1), y_1 \mapsto y\} \\
    \theta_2 &= \{x_1 \mapsto \ssym(x_2), y_2 \mapsto \ssym(y)\} \\
    \theta_3 &= \{x_2 \mapsto \ssym(x_3), y_3 \mapsto \ssym^2(y)\} \\
    & \vdots
\end{align*}
This rewrite sequence corresponds to the infinite construction
of the overlap closure of $\cR$. Indeed, for all $k > 0$, the pair
$(a_0\theta_1\dots\theta_k, a_k\theta_1\dots\theta_k)$ corresponds to
the rule $\big(\fsym(\ssym^k(x),y), \fsym(x,\ssym^k(y))\big)$:
\begin{align*}
    (a_0\theta_1, a_1\theta_1) &= (\fsym(\ssym(x_1),y), \fsym(x_1,\ssym(y))) \\
    (a_0\theta_1\theta_2, a_2\theta_1\theta_2) &=
    (\fsym(\ssym^2(x_2),y), \fsym(x_2,\ssym^2(y))) \\
    & \vdots
\end{align*}

\subsection{Loops}\label{sect:intro-loops}
The vast majority of the papers dealing with non-termination in
term rewriting provide conditions for the existence of \emph{loops}.
While sufficient conditions are generally used to design loop-detection
approaches~\cite{waldmann04,gieslTS05,zantema05,zanklM07,payet08},
necessary conditions are rather used to prove the absence of
loops~\cite{zantemaG96,geserZ99}.
In this context, a loop refers to a finite rewrite sequence
$a_0 \ra_{\Pi} \cdots \ra_{\Pi} a_n$ where $a_n$ embeds an instance
of $a_0$. It is well known that then, one can go on from $a_n$, \ie
there is a finite rewrite sequence $a_n \ra_{\Pi} \cdots \ra_{\Pi} a_{n+m}$
where $a_{n+m}$ embeds an instance of $a_n$, hence an instance of $a_0$,
and so on. Therefore, there is an infinite rewrite sequence starting
from $a_0$ and in which $a_0$ is ``reached'' periodically
(see Example~\ref{ex:trs-intro-2} below).

The situation in logic programming is different, in the sense that
there is no ``official'' definition of loops (at least, as far as
we know). In contrast to term rewriting, where loopingness is a
special form of non-termination, logic programming loops simply
seem to correspond to infinite rewrite sequences, with no more
precision. In this context, \cite{schreyeVB90,payetM06} provide
sufficient conditions for the existence of loops,
while~\cite{schreyeVB90,sahlin90,bolAK91,shen97} provide necessary
conditions. The former are generally used to design static
(\ie compile-time) loop-detection approaches, while the latter
are used to design dynamic (\ie runtime) loop checks (\ie rewrite
sequences are pruned at runtime when it seems appropriate, with a
risk of pruning a finite rewrite).

In this paper, we provide a generic definition of loops in any
ARS $\ars$. It generalises TRS loops by considering any
\emph{compatible} relation $\phi$, not necessarily the
``embeds an instance of'' relation (see Definition~\ref{def:loop}).
Compatibility means that if $a \rra_{\Pi} a_1$ and $a'$ is related
to $a$ via $\phi$ (denoted as $a' \in \phi(a)$) then we have
$a' \rra_{\Pi} a'_1$ for some $a'_1 \in \phi(a_1)$
(see Definition~\ref{def:compatibility}).

\begin{example}\label{ex:trs-intro-2}
    In Example~\ref{ex:trs-intro}, we have the rewrite sequence
    \begin{equation}\label{eq:ex-trs-intro}
        \underbrace{\fsym(x)}_{a_0} \ra_{(r_1,p_0)}
        \underbrace{\gsym(\hsym(x,\one),x)}_{a_1} \ra_{(r_2,p_1)}
        \underbrace{\gsym(\hsym(x,\zero),x)}_{a_2} \ra_{(r_3,p_2)}
        \underbrace{\gsym(\fsym^2(x),x)}_{a_3}
    \end{equation}
    where $p_0=\epsilon$ (\ie the \emph{root position}, which 
    is the position of a term inside itself), $p_1$ is the position
    of $\one$ in $a_1$ and $p_2$ is the position of $\hsym(x,\zero)$
    in $a_2$.
    Consider $\phi$ that relates any term $s$ to the terms that 
    embed an instance of $s$, \ie of the form $c[s\theta]$, where
    $c$ is a context (a term with a hole in it), $\theta$ a
    substitution and $c[s\theta]$ the term resulting from
    replacing the hole in $c$ by $s\theta$.
    It is well known that $\phi$ is compatible with the rewrite
    relation of term rewriting, \ie $s' \in \phi(s)$
    and $s \ra_{(r,p)} t$ imply that $s' \ra_{(r,p')} t'$
    for some position $p'$ and some $t' \in \phi(t)$.
    We have $a_3 \in \phi(a_0)$ because $a_3 = c[a_0\theta]$
    for $\theta = \{x \mapsto \fsym(x)\}$ and $c = \gsym(\square,x)$,
    where $\square$ is the hole. So, \eqref{eq:ex-trs-intro} is a
    loop and, by the compatibility property, there is an infinite
    rewrite sequence
    \[a_0 \ra_{(r_1,p_0)} a_1 \ra_{(r_2,p_1)} a_2 \ra_{(r_3,p_2)}
    \underbrace{a_3}_{\in \phi(a_0)} 
    \ra_{(r_1,p_3)} \underbrace{a_4}_{\in \phi(a_1)}
    \ra_{(r_2,p_4)} \underbrace{a_5}_{\in \phi(a_2)}
    \ra_{(r_3,p_5)} \underbrace{a_6}_{\in \phi(a_3)}
    \ra_{(r_1,p_6)} \cdots\]
    which is a succession of loops.
\end{example}

\begin{example}\label{ex:lp-intro-2}
    In Example~\ref{ex:lp-intro}, we have the rewrite sequence
    \begin{equation}\label{eq:ex-lp-intro}
        \underbrace{\sequence{\psym(\fsym(x,\zero))}}_{\seqset{a_0}}
        \lpa_{(r,\sequence{1})}
        \underbrace{\sequence{\psym(x),\qsym(x)}}_{\seqset{a_1}}
    \end{equation}
    Consider $\phi$ that relates any goal $\seqset{s}$ to
    the goals that embed a more general goal than
    $\seqset{s}$, \ie of the form $\seqset{c}[\seqset{t}]$,
    where $\seqset{c}$ is a context (a goal with a hole in it)
    and $\seqset{t}$ is more general than $\seqset{s}$, \ie
    $\seqset{s} = \seqset{t}\theta$ for some substitution
    $\theta$ (hence, $\seqset{s}$ is an instance of
    $\seqset{t}$). We prove in this paper
    (Lemma~\ref{lem:ins-mg-compatible}) that $\phi$ is
    compatible with the rewrite relation of logic programming,
    \ie $\seqset{s'} \in \phi(\seqset{s})$ and
    $\seqset{s} \lpa_{(r,p)} \seqset{t}$
    imply that  $\seqset{s'} \lpa_{(r,p')} \seqset{t'}$
    for some position $p'$ and some
    $\seqset{t'} \in \phi(\seqset{t})$.
    We have $\seqset{a_1} \in \phi(\seqset{a_0})$
    because $\seqset{a_1} = \seqset{c}[\seqset{b_0}]$ where
    $\seqset{c} = \sequence{\square,\qsym(x)}$
    and $\seqset{b_0} = \sequence{\psym(x)}$ is more general
    than $\seqset{a_0}$ (we have
    $\seqset{a_0} = \seqset{b_0}\{x \mapsto \fsym(x,\zero)\}$).
    So, \eqref{eq:ex-lp-intro} is a loop and, by the
    compatibility property, there is an infinite rewrite
    sequence
    \[\seqset{a_0} \lpa_{(r,\sequence{1})}
    \underbrace{\seqset{a_1}}_{\in \phi(\seqset{a_0})}
    \lpa_{(r,p_1)} \underbrace{\seqset{a_2}}_{\in \phi(\seqset{a_1})}
    \lpa_{(r,p_2)} \underbrace{\seqset{a_3}}_{\in \phi(\seqset{a_2})}
    \lpa_{(r,p_3)} \cdots\]
    which is a succession of loops.
\end{example}

\subsection{Binary Chains}
\label{sect:intro-binary-chains}
The infinite rewrite sequence of Example~\ref{ex:trs-intro-2}
(resp. Example~\ref{ex:lp-intro-2}) is a succession of loops
that all involve the same sequence $\sequence{r_1,r_2,r_3}$
(resp. $\sequence{r}$) of rules. In this paper, we also consider
infinite rewrite sequences that involve two sequences of rules.
We call them \emph{binary chains}
(see Definition~\ref{def:binary-chain}).

\begin{example}\label{ex:trs-binary-intro}
    Consider the contexts $c = \gsym(\square,\zero,\square)$
    and $c' = \gsym(\square,\one,\square)$ and the TRS that
    consists of the following rules:
    \[\begin{array}{l@{\qquad}l}
        r_1 = \big(\fsym(x, c[y],x), \hsym(x, y)\big) &
        r_3 = \big(\fsym(x, \zero, x), \fsym(c[x], c'[x], c[x])\big) \\[1ex]
        r_2 = \big(\hsym(x, y), \fsym(c[x], y, c[x])\big) &
        r_4 = \big(\one, \zero\big)
    \end{array}\]
    We have the infinite rewrite sequence
    \begin{align*}
        \underbrace{\fsym\big(c[\zero],c[\zero],c[\zero]\big)}_{a_0}
        & \ra_{(r_1,\epsilon)}
        \hsym\big(c[\zero],\zero\big)
        \ra_{(r_2,\epsilon)}
        \fsym\big(c^2[\zero],\zero,c^2[\zero]\big) \\
        & \ra_{(r_3,\epsilon)} 
        \underbrace{\fsym\big(c^3[\zero], c'[c^2[\zero]], c^3[\zero]\big)}_{a'_0}
        \ra_{(r_4,p_0)}
        \underbrace{\fsym\big(c^3[\zero], c^3[\zero], c^3[\zero]\big)}_{a_1} \\
        & \ra_{(r_1,\epsilon)}
        \hsym\big(c^3[\zero],c^2[\zero]\big)
        \ra_{(r_2,\epsilon)}
        \fsym\big(c^4[\zero],c^2[\zero],c^4[\zero]\big) \\
        & \ra_{(r_1,\epsilon)}
        \hsym\big(c^4[\zero],c[\zero]\big)
        \ra_{(r_2,\epsilon)}
        \fsym\big(c^5[\zero],c[\zero],c^5[\zero]\big) \\
        & \ra_{(r_1,\epsilon)}
        \hsym\big(c^5[\zero],\zero\big)
        \ra_{(r_2,\epsilon)}
        \fsym\big(c^6[\zero],\zero,c^6[\zero]\big) \\
        & \ra_{(r_3,\epsilon)}
        \underbrace{\fsym\big(c^7[\zero], c'[c^6[\zero]], c^7[\zero]\big)}_{a'_1}
        \ra_{(r_4,p_1)}
        \underbrace{\fsym\big(c^7[\zero], c^7[\zero], c^7[\zero]\big)}_{a_2} \\
        & \ra_{(r_1,\epsilon)} \cdots
    \end{align*}
    where $c^n$ denotes $n$ embeddings of $c$ into itself and $p_n$
    is the position of $\one$ (occurring in $c'$) in $a'_n$, for all
    $n$ in the set $\nat$ of non-negative integers. We note that for all
    $n \in \nat$, from $a_n$ to $a_{n+1}$ the sequence
    $w_1 = \sequence{r_1,r_2}$ is applied several times and then the
    sequence $w_2 = \sequence{r_3,r_4}$ is applied exactly once.
    This is written as
    $a_n \mathop{\left(\ra^*_{w_1} \circ \ra_{w_2}\right)} a_{n+1}$,
    and hence the infinite rewrite above relies on the two sequences
    of rules $w_1$ and $w_2$.
    In Section~\ref{sect:rec-pairs}, we present a syntactic criterion 
    for detecting a special case of this form of non-termination.
\end{example}

\subsection{Implementation}
\label{sect:intro-simplified-setting}
Our tool~\nti{} (Non-Termination Inference)~\cite{nti}
performs automatic proofs of non-termination in
term rewriting and logic programming. It specifically
considers logic programming with Prolog's leftmost selection
rule, \ie always the leftmost term is chosen in a goal for
narrowing (contrary to Example~\ref{ex:lp-intro}, where
$s_1$ is not the leftmost term of $\seqset{s}$).
It applies program transformation techniques as the very
first step, namely, the
\emph{dependency pair unfolding}~\cite{payet18} in
term rewriting and the
\emph{binary unfolding}~\cite{gabbrielliG94,codishT99}
in logic programming. The idea is to compute ``compressed''
forms of finite rewrite sequences that can be used in a
simplified setting not polluted by technicalities, such as
deeper and deeper embedding of terms
(\eg in Example~\ref{ex:trs-intro-2},
$a_3$ embeds an instance of $a_0$, 
$a_6$ embeds an instance of $a_3$, \dots)
This allows us to detect loops using a simplified version
of the relations $\phi$ of Examples~\ref{ex:trs-intro-2}
and~\ref{ex:lp-intro-2}, and also to concentrate on rewrite
sequences consisting of only one step (as they are
``compressions'' of longer sequences),
see Examples~\ref{ex:trs-intro-simplified},
\ref{ex:trs-binary-intro-simplified}
and~\ref{ex:lp-intro-simplified} below.

\subsubsection{Dependency Pair Unfolding}
The defined symbols of a TRS $\cR$ are the function
symbols $\fsym$ for which there is a rule of the form
$(\fsym(\dots),\dots)$ in $\cR$. Each defined symbol $\fsym$
is associated with a symbol $\Fsym$ of the same arity that
does not occur in $\cR$. The set of dependency pairs of $\cR$ is
\[\left\{\big(\Fsym(s_1,\dots,s_n),\Gsym(t_1,\dots,t_m)\big)
\;\middle\vert\; 
\begin{array}{l}
    \big(\fsym(s_1,\dots,s_n),t\big) \in \cR \\
    \gsym(t_1,\dots,t_m) \text{ is a subterm of } t \\
    \gsym \text{ is a defined symbol of $\cR$}
\end{array}
\right\}\]

The technique of~\cite{payet18} transforms a TRS $\cR$ into
a (possibly infinite) set $\unf(\cR)$ inductively defined as
follows. It exactly consists of the dependency pairs of $\cR$,
together with the rules resulting from narrowing (forwards and
backwards, as for $\oc(\cR)$) elements of $\unf(\cR)$ with
rules and dependency pairs of $\cR$; while dependency pairs
are only applied at the root position, between the application
of two dependency pairs there can be an arbitrary number
of narrowing steps below the root, using rules of $\cR$.
The computation of $\unf(\cR)$ is hence very similar to
that of $\oc(\cR)$, except that: it also considers the
dependency pairs, it only uses the initial dependency pairs
and rules for narrowing elements of $\unf(\cR)$ and it
allows narrowing of variable subterms.

\begin{theorem}[\cite{payet18}]\label{thm:term-criterion-trs}
    Let $\cR$ be a TRS. If a term $\Fsym(s_1,\dots,s_n)$
    starts an infinite rewrite sequence \wrt{} $\unf(\cR)$
    then the term $\fsym(s_1,\dots,s_n)$
    starts an infinite rewrite sequence \wrt{} $\cR$.
\end{theorem}

Note that this theorem only states an implication, not an
equivalence, contrary to Theorem~\ref{thm:observing-termination-lp}
below in logic programming (where an ``if and only if'' is stated).
Therefore, we do not know whether it is restrictive to detect
infinite rewrite sequences \wrt{} $\unf(\cR)$, instead of $\cR$
(\ie we do not know whether the existence of an infinite rewrite
sequence \wrt{} $\cR$ necessarily implies that of an infinite
rewrite sequence \wrt{} $\unf(\cR)$).

\begin{example}[Example~\ref{ex:trs-intro-2} continued]
    \label{ex:trs-intro-simplified}
    Let $\cR$ be the TRS under consideration.
    The following rules are dependency pairs of $\cR$ obtained
    from $r_1$ and $r_3$ respectively:
    \[r'_1 = \big(\Fsym(x), \Hsym(x,\one)\big) = (u'_1, v'_1)
    \qquad
    r'_3 = \big(\Hsym(x,\zero), \Fsym(\fsym(x))\big)\]
    Here, $r'_1$ allows us to get rid of the context
    $c = \gsym(\square,x)$ that occurs in $r_1$.
    Let us simplify $\phi$ into the relation $\phi'$ that binds
    any term to its instances (so, compared to $\phi$,
    there are no more contexts). Of course, $\phi'$ is also
    compatible with the rewrite relation of term rewriting.
    We have $v'_1 \lpa_{(r_2,p')} \Hsym(x,\zero)$ where $p'$
    is the position of $\one$ in $v'_1$; so, from $r'_1$
    and $r_2$, the technique of \cite{payet18} produces
    $r''_1 = \big(\Fsym(x), \Hsym(x,\zero)\big) = (u''_1,v''_1)$.
    Moreover, $v''_1 \lpa_{(r'_3,\epsilon)} \Fsym(\fsym(x))$ so,
    from $r''_1$ and $r'_3$, we get
    $r = \big(\Fsym(x), \Fsym(\fsym(x))\big)$.
    Let $\pi = (r,\epsilon)$.
    We have the following loop \wrt{} $\phi'$:
    \[\overbrace{\Fsym(x)}^{a'_0} \ra_{\pi}
    \overbrace{\Fsym(\fsym(x))}^{a'_1 \in \phi'(a'_0)}\]
    Hence, there is an infinite rewrite sequence 
    $a'_0 \ra_{\pi} a'_1 \ra_{\pi} a'_2 \ra_{\pi} \cdots$
    with $a'_1 \in \phi'(a'_0)$, $a'_2 \in \phi'(a'_1)$, \dots{}
    It corresponds to that of Example~\ref{ex:trs-intro-2},
    but each successive application of $r_1,r_2,r_3$
    has been ``compressed'' into a single application of $r$.
    Moreover, it does not involve $c$, as well as the context
    that embeds an instance of $a_1$ in $a_4$, \emph{etc}.
    By Theorem~\ref{thm:term-criterion-trs}, $\fsym(x)$ starts
    an infinite rewrite sequence \wrt{} $\cR$.
\end{example}

\begin{example}\label{ex:trs-binary-intro-simplified}
    The following dependency pairs are obtained respectively
    from the rules $r_1$, $r_2$ and $r_3$ of the TRS $\cR$ of
    Example~\ref{ex:trs-binary-intro}:
    \begin{align*}
        r'_1 &= \big(\Fsym(x, c[y], x), \Hsym(x, y)\big) \\
        r'_2 &= \big(\Hsym(x, y), \Fsym(c[x], y, c[x])\big) \\
        r'_3 &= \big(\Fsym(x, \zero, x), \Fsym(c[x], c'[x], c[x])\big)
    \end{align*}
    From $r'_1$ and $r'_2$, the technique of~\cite{payet18} produces 
    $r''_1 = \big(\Fsym(x, c[y], x), \Fsym(c[x], y, c[x])\big)$
    and, from $r'_3$ and $r_4$, it produces 
    $r''_3 = \big(\Fsym(x, \zero, x), \Fsym(c[x], c[x], c[x])\big)$.
    Let $\pi_1 = (r''_1,\epsilon)$, $\pi_3 = (r''_3,\epsilon)$
    and $A_n = \Fsym\big(c^n[\zero],c^n[\zero],c^n[\zero]\big)$
    for all $n \in \nat$. We have the binary chain:
    \begin{align*}
        A_1 \ra_{\pi_1}
        \Fsym\big(c^2[\zero],\zero,c^2[\zero]\big) \ra_{\pi_3} 
        A_3 \ra^3_{\pi_1}
        \Fsym\big(c^6[\zero],\zero,c^6[\zero]\big) \ra_{\pi_3}
        A_7 \ra^7_{\pi_1} \cdots
    \end{align*}
    It corresponds to that of Example~\ref{ex:trs-binary-intro},
    but each application of $w_1$ (resp. $w_2$) has been ``compressed''
    into an application of $r''_1$ (resp. $r''_3$).
    By Theorem~\ref{thm:term-criterion-trs},
    $\fsym\big(c[\zero],c[\zero],c[\zero]\big)$ starts
    an infinite rewrite sequence \wrt{} $\cR$.
\end{example}

\subsubsection{Binary Unfolding}
The binary unfolding~\cite{gabbrielliG94,codishT99} is
a program transformation technique that has been introduced
in the context of Prolog's leftmost selection rule.
It transforms a LP $P$ into a (possibly infinite) set
$\binunf(P)$ of \emph{binary} rules (\ie rules, the 
right-hand side of which is a goal that contains at most
one term) inductively defined as follows: $\binunf(P)$
exactly consists of the rules constructed by narrowing
prefixes of right-hand sides of rules of $P$ using
elements of $\binunf(P)$; more precisely, for all
$(u,\sequence{v_1,\dots,v_n}) \in P$:
\begin{enumerate}[label=(\Alph*)]
    \item \label{enum-intro-1}
    for each $1 \leq i \leq n$, the goals
    $\sequence{v_1},\dots,\sequence{v_{i-1}}$ are narrowed,
    respectively, with $(u_1,\epsilon),\dots,(u_{i-1},\epsilon)$
    from $\binunf(P)$ (where $\epsilon$ is the empty goal)
    to obtain a corresponding instance of
    $(u,\sequence{v_i})$,
    \item for each $1 \leq i \leq n$, the goals
    $\sequence{v_1},\dots,\sequence{v_{i-1}}$ are narrowed,
    respectively, with $(u_1,\epsilon),\dots,(u_{i-1},\epsilon)$
    from $\binunf(P)$ and $\sequence{v_i}$ is also narrowed with
    $(u_i,\sequence{v})$ from $\binunf(P)$ to obtain
    a corresponding instance of $(u,\sequence{v})$,
    \item $\sequence{v_1},\dots,\sequence{v_n}$ are narrowed,
    respectively, with $(u_1,\epsilon),\dots,(u_n,\epsilon)$
    from $\binunf(P)$ to obtain a corresponding instance of
    $(u,\epsilon)$.
\end{enumerate}
Intuitively, each generated binary rule $(u,\sequence{v})$
specifies that, \wrt{} the leftmost selection rule,
$\sequence{u}$ can be rewritten, using rules of $P$, to a
goal $\sequence{v,\dots}$. 
\begin{theorem}[\cite{codishT99}]
    \label{thm:observing-termination-lp}
    Let $P$ be a LP and $\seqset{s_0}$ be a goal.
    Assume that the leftmost selection rule is used. Then,
    $\seqset{s_0}$ starts an infinite rewrite sequence \wrt{} $P$
    \emph{if and only if} it starts an infinite rewrite
    sequence \wrt{} $\binunf(P)$.
\end{theorem}

So, in this context, suppose that $\seqset{s_0}$
has the form $\sequence{v_0,\dots,v_n}$ and that it
starts an infinite rewrite sequence \wrt{} $P$.
Then, by Theorem~\ref{thm:observing-termination-lp},
there is an infinite rewrite sequence 
\begin{equation}\label{eq:ex-observing-termination-lp}
    \seqset{s_0} \lpa^{\theta_0}_{(r_0,p_0)}
    \seqset{s_1} \lpa^{\theta_1}_{(r_1,p_1)}
    \cdots
\end{equation}
where $r_0,r_1,\dots$ are rules of $\binunf(P)$
(we decorate each step with the corresponding unifier).
As the leftmost selection rule is used, $p_i$ is the
position of the leftmost term of $\seqset{s_i}$, for all
$i \in \nat$. So, necessarily, there is a finite
(possibly empty) prefix
of~\eqref{eq:ex-observing-termination-lp} that ``erases''
a (possibly empty) prefix of $\seqset{s_0}$, \ie it has
the form $\seqset{s_0} \lpa^{\theta_0}_{(r_0,p_0)} \cdots
\lpa^{\theta_{m-1}}_{(r_{m-1},p_{m-1})} \seqset{s_m}$
with $\seqset{s_m} = \sequence{v_k,\dots,v_n}\theta_0\dots\theta_{m-1}$
and $\seqset{t_m} = \sequence{v_k\theta_0\dots\theta_{m-1}}$
starts an infinite rewrite sequence of the form
$\seqset{t_m} \lpa_{(r_m,p'_m)} \seqset{t_{m+1}}
\lpa_{(r_{m+1},p'_{m+1})} \cdots$.
As the goal $\seqset{t_m}$ is a singleton and the rules
$r_m,r_{m+1},\dots$ are binary, we necessarily have that
$\seqset{t_{m+1}},\seqset{t_{m+2}},\dots$ are
singletons and $p'_m = p'_{m+1} = \cdots = \sequence{1}$
(because, in logic programming, rewriting only takes place
at the root position of a term of a goal).

So, from the non-termination point of view, and as far as
Prolog's leftmost selection rule is concerned, it is not
restrictive to consider a simplified form of logic
programming where one only rewrites singleton goals using
binary rules. Indeed, we have shown above that if $P$ is
non-terminating in the full setting then $\binunf(P)$ is
necessarily non-terminating in this simplified one.
Moreover, the leftmost selection rule is the most used by far
(\eg it is the only considered one in the Termination
Competition~\cite{termcomp}).
In Section~\ref{sect:ars}, for all LPs $P$, we define a
rewrite relation $\hra_P$ that models this simplified
form of logic programming. It enjoys an interesting closure
property (Lemma~\ref{lem:closed-subs-lp}) that we will use
to detect binary chains (Section~\ref{sect:rec-pairs}).

\begin{example}\label{ex:lp-intro-simplified}
    In Example~\ref{ex:lp-intro-2}, the binary unfolding of
    the LP $P$ under consideration produces the rule 
    $r' = \big(\psym(\fsym(x,\zero)), \sequence{\psym(x)}\big)$
    from $r$ (using~\ref{enum-intro-1} in the definition
    of $\binunf(P)$ above, with $i = 1$). So, compared to $r$,
    we do not have the context
    $\seqset{c} = \sequence{\square,\qsym(x)}$ anymore and we can
    simplify $\phi$, \ie we consider $\phi'$ that relates any goal
    $\seqset{s}$ to all the goals that are more general than
    $\seqset{s}$. Of course, $\phi'$ is also compatible with
    the rewrite relation of LPs. Let $\pi = (r',\sequence{1})$.
    We have this loop \wrt{} $\phi'$:
    \[\underbrace{\sequence{\psym(\fsym(x,\zero))}}_{\seqset{a_0}}
    \lpa_{\pi}
    \underbrace{\sequence{\psym(x)}}_{\seqset{b_0} \in \phi'(\seqset{a_0})}\]
    Hence, there is an infinite rewrite sequence 
    $\seqset{a_0} \lpa_{\pi} \seqset{b_0}
    \lpa_{\pi} \seqset{b_1} \lpa_{\pi} \cdots$
    where $\seqset{b_0} \in \phi'(\seqset{a_0})$,
    $\seqset{b_1} \in \phi'(\seqset{b_0})$, \dots{}
    It corresponds to that of Example~\ref{ex:lp-intro-2}, but
    it does not involve the context that embeds a more general goal
    than $\seqset{a_0}$ in $\seqset{a_1}$, the context that embeds
    a more general goal than $\seqset{a_1}$ in $\seqset{a_2}$,
    \dots; moreover, as $\seqset{a_0}, \seqset{b_0}, \seqset{b_1}, \dots$
    consist of only one term, the narrowing steps are all performed at
    position $\sequence{1}$. By Theorem~\ref{thm:observing-termination-lp},
    $\seqset{a_0}$ starts an infinite rewrite sequence \wrt{} $P$.
\end{example}

\subsection{Content of the Paper}
After presenting some preliminary material in
Section~\ref{sect:prel} and abstract reduction systems
in Section~\ref{sect:ars}, we describe, in a generic way,
loops in Section~\ref{sect:loops} and binary chains in
Section~\ref{sect:binary-chains}, together with an
automatable approach for the detection of the latter.
Then, in Section~\ref{sect:implementation}, we compare 
our implementation \nti{} with the tools that participated
in the Termination Competition~2023~\cite{termcomp}.
Finally, we present related work in Section~\ref{sect:rel-work}
and conclude in Section~\ref{sect:conclusion}.

\section{Preliminaries}
\label{sect:prel}
We introduce some basic notations and definitions.
First, in Section~\ref{sect:finite-sequence}, we
consider finite sequences and, in
Section~\ref{sect:binary-relation},
binary relations and the related notions of
chain and compatibility.    
Then, in Sections~\ref{sect:term} and~\ref{sect:subs},
we define terms, contexts, substitutions and
unifiers. Finally, in Section~\ref{sect:trs-lp}, we
present term rewriting and logic programming.

We let $\nat$ denote the set of non-negative integers.

\subsection{Finite Sequences}\label{sect:finite-sequence}
Let $A$ be a set.
We let $\seqset{A}$ denote the set of finite sequences
of elements $A$; it includes the empty sequence, denoted
as $\epsilon$. We generally (but not always) use the 
delimiters $\sequence{\cdot}$ for writing elements of
$\seqset{A}$.
Moreover, we use juxtaposition to denote the concatenation
operation, \eg $\sequence{a_0,a_1}\sequence{a_2,a_3} =
\sequence{a_0,a_1,a_2,a_3}$  and
$a_0\sequence{a_1,a_2} = \sequence{a_0,a_1,a_2}$.
The length of $w \in \seqset{A}$ is denoted as $\length{w}$
and is inductively defined as:  $\length{w} = 0$ if
$w = \epsilon$ and $\length{w} = 1 + \length{w'}$ if
$w = aw'$ for some $a \in A$ and $w' \in \seqset{A}$.

\subsection{Binary Relations}\label{sect:binary-relation}
A \emph{binary relation} $\psi$ from a set $A$ to a set $B$
is a subset of $A \times B$. For all $a \in A$, we let
$\psi(a) = \{b \in B \mid (a,b) \in \psi\}$.
Instead of $(a,b) \in \psi$, we may write $a \mathop{\psi} b$
(\eg for binary relations that have the form of an arrow) or
$b \in \psi(a)$. We let
$\psi^{-1} = \{(b,a) \in B \times A \mid (a,b) \in \psi\}$
be the \emph{inverse} of $\psi$.
A \emph{function} $f$ from $A$ to $B$ is a binary relation
from $A$ to $B$ which is such that for all
$a \in A$, $f(a)$ consists of exactly one element.

A binary relation $\phi$ on a set $A$ is a subset of
$A^2 = A \times A$. For all $\varphi \subseteq A^2$, we let
$\phi \circ \varphi$ denote the \emph{composition} of
$\phi$ and $\varphi$:
\[\phi \circ \varphi = \left\{(a,a') \in A^2 \;\middle\vert\;
\exists a_1 \in A \ (a,a_1) \in \phi \land
(a_1,a') \in \varphi \right\}\]
We let $\phi^0$ be the identity relation and,
for any $n\in\nat$, $\phi^{n+1}=\phi^n \circ \phi$.
Moreover, $\phi^+ = \cup \{\phi^n \mid n > 0 \}$
(resp. $\phi^* = \phi^0 \mathop{\cup} \phi^+$) denotes
the transitive (resp. reflexive and transitive) \emph{closure}
of $\phi$.
\begin{definition}\label{def:chain}
  Let $\phi$ be a binary relation on a set $A$.
  A \emph{$\phi$-chain}, or simply \emph{chain} if $\phi$ is
  clear from the context, is a (possibly infinite) sequence
  $a_0, a_1, \dots$ of elements of $A$ such that
  $a_{n+1} \in \phi(a_n)$ for all  $n \in \nat$. For
  binary relations that have the form of an arrow, \eg $\rra$,
  we simply write it as $a_0 \rra a_1 \rra \cdots$.
\end{definition}

The next definition resembles that of a
\emph{simulation relation}~\cite{park81} in a
state transition system: $\rra$ corresponds to
the transition relation of the system and $\phi$
to the simulation.
\begin{definition}\label{def:compatibility}
  Let $A$ be a set. We say that $\rra \subseteq A^2$ and
  $\phi \subseteq A^2$ are \emph{compatible} if for all
  $a,a',a_1 \in A$ there exists $a'_1 \in A$ such that
  \[\left(a' \in \phi(a) \land a \rra a_1\right)
  \text{ implies }
  \left(a'_1 \in \phi(a_1) \land a' \rra a'_1\right)\]
  Equivalently, $(\phi^{-1} \circ {\rra})
  \subseteq ({\rra} \circ \phi^{-1})$.
  This is illustrated by the following diagram:
  \begin{center}
    \begin{tikzcd}
      a \arrow[r, "\rra"] \arrow[d, "\phi"'] & a_1 \arrow[dashed, d, "\phi"] \\
      a' \arrow[dashed, r, "\rra"] & a'_1
    \end{tikzcd}
  \end{center}
  (solid (resp. dashed) arrows represent universal
  (resp. existential) quantification).
\end{definition}

The following result is straightforward.
\begin{lemma}\label{lem:compatibility-composition}
  Let $\rra$, $\hookrightarrow$ and $\phi$ be binary relations
  on a set $A$. If $\rra$ and $\phi$ are compatible and
  $\hookrightarrow$ and $\phi$ are compatible then
  $(\rra \circ \hookrightarrow)$ and $\phi$ are compatible.
\end{lemma}
\begin{proof}
  Using the same notations as in
  Definition~\ref{def:compatibility},
  we have:
  \begin{center}
    \begin{tikzcd}
      a \arrow[r, "\rra"] \arrow[d, "\phi"'] &
      a_1 \arrow[r, "\hookrightarrow"] \arrow[dashed, d, "\phi"] &
      a_2 \arrow[dashed, d, "\phi"] \\
      a' \arrow[dashed, r, "\rra"] &
      a'_1 \arrow[dashed, r, "\hookrightarrow"] & a'_2
    \end{tikzcd}
  \end{center}
\end{proof}

\subsection{Terms}\label{sect:term}
We use the same definitions and notations as~\cite{baaderN98} for terms.
\begin{definition}
  A \emph{signature} $\Sigma$ is a set of \emph{function symbols},
  each element of which has an \emph{arity} in $\nat$, which
  is the number of its arguments. For all $n\in\nat$, we denote the
  set of all $n$-ary elements of $\Sigma$ by $\Sigma^{(n)}$. The
  elements of $\Sigma^{(0)}$ are called \emph{constant symbols}.
\end{definition}

Function symbols are denoted by letters or non-negative integers
in the \emph{sans serif} font, \eg $\fsym,\gsym$ or $\zero,\one$.
We frequently use the superscript notation to denote several
successive applications of a unary function symbol,
\eg $\ssym^3(\zero)$ is a shortcut for
$\ssym(\ssym(\ssym(\zero)))$ and $\ssym^0(\zero) = \zero$.

\begin{definition}
  Let $\Sigma$ be a signature and $X$ be a set of \emph{variables}
  disjoint from $\Sigma$. The set $\termset$ is
  defined as:
  \begin{itemize}
    \item $X \subseteq \termset$,
    \item for all $n\in\nat$, all $\fsym \in \Sigma^{(n)}$ and
    all $s_1,\dots,s_n \in \termset$, 
    $\fsym(s_1,\dots,s_n) \in \termset$.
  \end{itemize}
  For all $s \in \termset$, we let $\vars(s)$ denote the set of
  variables occurring in $s$. Moreover, for all
  $\sequence{s_1,\dots,s_n} \in \seqset{\termset}$, we let
  $\vars(\sequence{s_1,\dots,s_n}) =
  \vars(s_1) \cup \cdots \cup \vars(s_n)$.
\end{definition}

In order to simplify the statement of the definitions and theorems
of this paper, from now on we fix a signature $\Sigma$, an infinite
countable set $X$ of variables disjoint from $\Sigma$ and a constant
symbol $\square$ which does not occur in $\Sigma \cup X$.
\begin{definition}\label{def:term-context}
  A \emph{term} is an element of $\termset$ and 
  a \emph{goal} an element of $\seqset{\termset}$.
  Moreover, a \emph{context} is an element of $\termsett$ that
  contains at least one occurrence of $\square$
  and a \emph{goal-context} is a finite sequence
  of the form
  $\sequence{s_1,\dots,s_i,\square,s_{i+1},\dots,s_n}$
  where all the $s_i$'s are terms. 
  A context $c$ can be seen as a term with ``holes'', represented
  by $\square$, in it. For all $t \in \termsett$, we let $c[t]$
  denote the element of $\termsett$ obtained from $c$ by replacing
  all the occurrences of $\square$ by $t$. We use the superscript
  notation for denoting several successive embeddings of a context 
  into itself: $c^0 = \square$ and, for all $n \in \nat$,
  $c^{n+1} = c[c^n]$.
  Identically, a goal-context $\seqset{c}$ can be seen as
  a goal with a hole, represented by $\square$. For all
  $\seqset{t} \in \seqset{\termset}$, we let
  $\seqset{c}[\seqset{t}]$ denote the goal obtained
  from $\seqset{c}$ by replacing $\square$ by the
  elements of $\seqset{t}$.
\end{definition}

Terms are generally denoted by $a,s,t,u,v$, variables by $x,y,z$
and contexts by $c$, possibly with subscripts and primes.
Goals and goal-contexts are denoted using an overbar.

The notion of position in a term is needed to
define the operational semantics of term rewriting 
(see Definition~\ref{def:trs-lp-rew-rel}).

\begin{definition}
  The set of \emph{positions} of $s\in \termsett$, denoted as
  $\pos(s)$, is a subset of $\seqset{\nat}$ which is inductively
  defined as:
  \begin{itemize}
    \item if $s \in X$, then
    $\pos(s) = \{\epsilon\}$,
    \item if $s = \fsym(s_1,\dots,s_n)$ then
    $\pos(s) = \{\epsilon\} \cup
    \bigcup_{i = 1}^n \{ip \mid p \in \pos(s_i)\}$.
  \end{itemize}
  The position $\epsilon$ is called the \emph{root position}
  of $s$ and the function or variable symbol at this position
  is the \emph{root symbol}.
  For all $p\in\pos(s)$, the \emph{subterm of $s$ at position $p$},
  denoted by $\subterm{s}{p}$ is inductively defined as:
  $\subterm{s}{\epsilon} = s$ and
  $\subterm{\fsym(s_1,\dots,s_n)}{ip'} = \subterm{s_i}{p'}$.
  Moreover, for all $t\in\termset$, we denote by $s[t]_p$ the term
  that is obtained from $s$ by replacing the subterm at position $p$
  by $t$, \ie $s[t]_{\epsilon} = t$ and
  $\fsym(s_1,\dots,s_n)[t]_{ip'} =
  \fsym(s_1,\dots,s_i[t]_{p'},\dots,s_n)$.
\end{definition}

\subsection{Substitutions}\label{sect:subs}
\begin{definition}
  The set $\subsset$ of all \emph{substitutions} consists
  of the functions $\theta$ from $X$ to $\termset$ such that
  $\theta(x) \neq x$ for only finitely many variables $x$.
  The \emph{domain} of $\theta$ is the finite set
  $\dom(\theta)=\{x \in X \mid \theta(x) \neq x\}$.
  We usually write $\theta$ as
  $\{x_1\mapsto\theta(x_1), \dots, x_n\mapsto\theta(x_n)\}$
  where $\{x_1,\dots,x_n\} = \dom(\theta)$.
  A \emph{(variable) renaming} is a substitution that is a
  bijection on $X$.
\end{definition}

The application of a substitution $\theta$ to 
a term or context $s$ is denoted as $s\theta$
and is defined as:
\begin{itemize}
  \item $s\theta = \theta(s)$ if $s\in X$,
  \item $s\theta = \fsym(s_1\theta,\dots,s_n\theta)$
  if $s = \fsym(s_1,\dots,s_n)$.
\end{itemize}
This is extended to goals, \ie 
$\sequence{s_1,\dots,s_n}\theta =
\sequence{s_1\theta,\dots,s_n\theta}$.

\begin{definition}\label{def:instance-variant}
  Let $s,t \in \termset$.
  We say that $t$ is an \emph{instance} of $s$ if
  $t=s\theta$ for some $\theta \in \subsset$. Then,
  we also say that $s$ is \emph{more general than} $t$.
  If $\theta$ is a renaming, then $t$ is also called
  a \emph{variant} of $s$.
  These definitions straightforwardly extend to all 
  $\seqset{s}, \seqset{t} \in \seqset{\termset}$.
\end{definition}

\begin{definition}%
  The \emph{composition} of substitutions $\sigma$ and
  $\theta$ is the substitution denoted as $\sigma\theta$
  and defined as: for all $x \in X$,
  $\sigma\theta(x) = (\sigma(x))\theta$.
  We say that $\sigma$ is more general than $\theta$ if
  $\theta=\sigma\eta$ for some substitution $\eta$.
\end{definition}

The composition of substitutions is an associative operation, \ie
for all terms $s$ and all substitutions $\sigma$ and $\theta$,
$(s\sigma)\theta=s(\sigma\theta)$.
We use the superscript notation for denoting several successive
compositions of a substitution with itself: $\theta^0=\emptyset$
(the identity substitution) and, for all $n\in\nat$,
$\theta^{n+1}=\theta\theta^n=\theta^n\theta$.

\begin{definition}
  Let $s,t \in \termset$.
  We say that $s$ \emph{unifies} with $t$ (or that $s$ and $t$ unify)
  if $s\sigma=t\sigma$ for some $\sigma\in \subsset$. Then, $\sigma$
  is called a \emph{unifier} of $s$ and $t$ and $\mgu(s,t)$ denotes
  the \emph{most general unifier} (mgu) of $s$ and $t$, which is
  unique (up to variable renaming).
\end{definition}

We will frequently refer to the relations ``embeds an instance of''
(denoted as $\phiinst$) and ``embeds a more general term than''
(denoted as $\phimg$) defined as:
\begin{definition}\label{def:bin-rel}
  Using the usual notations for terms ($s,t$),
  goals ($\seqset{s},\seqset{t}$),
  contexts ($c$) and goal-contexts ($\seqset{c}$),
  we define:
  \begin{align*}
    \phiinst &= \left\{ \left(s,c[t]\right)
    \;\middle\vert\; t \text{ is an instance of } s
    \right\} \cup
    \left\{ \left(\seqset{s},\seqset{c}[\seqset{t}]\right)
    \;\middle\vert\;
    \seqset{t} \text{ is an instance of } \seqset{s}
    \right\} \\
    \phimg &= \left\{ \left(s,c[t]\right)
    \;\middle\vert\;
    t \text{ is more general than } s
    \right\} \cup
    \left\{ \left(\seqset{s},\seqset{c}[\seqset{t}]\right)
    \;\middle\vert\;
    \seqset{t} \text{ is more general than } \seqset{s}
    \right\}
  \end{align*}
\end{definition}

\subsection{Term Rewriting and Logic Programming}
\label{sect:trs-lp}
We refer to~\cite{baaderN98,terese03} for the basics of
term rewriting and to~\cite{lloyd87,apt97} for those
of logic programming.
For the sake of simplicity and harmonisation, we consider
the following notion of rule that encompasses TRS rules
and LP rules (usually, the right-hand side of a TRS rule
is a term, see Section~\ref{sect:intro-trs-lp}).

\begin{definition}\label{def:prog}
  A \emph{program} is a subset of
  $\termset \times \seqset{\termset}$,
  every element $(u,\seqset{v})$ of which is called
  a \emph{rule}. The term $u$ (resp. the (possibly
  empty) goal $\seqset{v}$) is the
  \emph{left-hand side} (resp. \emph{right-hand side}).
\end{definition}

The rules of a program allow one to rewrite terms and
goals. This is formalised by the following binary
relations, where $\ra_P$ corresponds to the
operational semantics of term rewriting and $\lpa_P$
to that of logic programming.
For all goals $\seqset{s}$ and rules
$(u,\seqset{v})$ and $(u',\seqset{v}')$, we write
$(u,\seqset{v}) \ll_{\seqset{s}} (u',\seqset{v}')$
to denote that $(u, \seqset{v})$ is a variant of
$(u',\seqset{v}')$ variable disjoint with $\seqset{s}$,
\ie for some renaming $\gamma$, we have
$u = u'\gamma$, $\seqset{v} = \seqset{v}'\gamma$ and
$\vars(u) \cap \vars(\seqset{s}) =
\vars(\seqset{v}) \cap \vars(\seqset{s}) = \emptyset$.
\begin{definition}\label{def:trs-lp-rew-rel}
  For all programs $P$, we let
  \[\ra_P = \bigcup \left\{\ra_r \;\middle\vert\; r\in P \right\}
  \quad\text{and}\quad
  \lpa_P = \bigcup \left\{\lpa_r \;\middle\vert\; r\in P \right\}\]
  where, for all $r \in P$,
  \[\ra_r = \bigcup \left\{ \ra_{(r,p)}
  \;\middle\vert\; p \in \seqset{\nat} \right\}
  \quad\text{and}\quad
  \lpa_r = \bigcup \left\{ \lpa_{(r,p)}
  \;\middle\vert\; p \in \seqset{\nat} \right\}
  \]
  and, for all $p \in \seqset{\nat}$,
  \begin{align*}
    \ra_{(r,p)} &= \left\{
      \big(s,t\big) \in \termset^2
      \;\middle\vert\;
      \begin{array}{l}
        r = (u, \sequence{v}), \ p \in \pos(s) \\
        \subterm{s}{p} = u\theta, \
        \theta \in \subsset \\
        t = s[v\theta]_p
      \end{array}
    \right\} \\
    \lpa_{(r,p)} &= \left\{
      \big(\seqset{s},\seqset{t}\big) \in \seqset{\termset}^2
      \;\middle\vert\;
      \begin{array}{l}
        \seqset{s} = \sequence{s_1,\dots,s_n}, \
        p \in \{\sequence{1},\dots,\sequence{n}\} \\
        (u, \sequence{v_1,\dots,v_m}) \ll_{\seqset{s}} r, \
        \theta = \mgu(s_p, u) \\
        \seqset{t} = \sequence{s_1,\dots,s_{p-1},v_1,
        \dots,v_m,s_{p-1},\dots,s_n}\theta
      \end{array}
    \right\}
  \end{align*}
\end{definition}
Concrete examples are provided in
Section~\ref{sect:intro-trs-lp}. We note that if the
rule $r$ has the form $(u,\seqset{v})$ with
$\length{\seqset{v}} \neq 1$ then
$\ra_{(r,p)} = \emptyset$ for all $p \in \seqset{\nat}$.
On the other hand, in the definition of $\lpa_{(r,p)}$,
we may have $\sequence{v_1,\dots,v_m} = \epsilon$
(\ie $m<1$). 

The next lemma states some closure properties of
$\ra_{(r,p)}$ and $\lpa_{(r,p)}$ under substitutions.
It is needed to prove compatibility of $\ra_r$ and
$\phimg$ (Lemma~\ref{lem:ins-mg-compatible} below)
as well as closure properties of abstract reduction
systems (see Section~\ref{sect:closure-subs}).
As $\ra_{(r,p)}$ relies on instantiation, its closure
property is almost straightforward. In contrast,
$\lpa_{(r,p)}$ relies on narrowing, so its closure
property is more restricted and, moreover, it is more
complicated to prove.

\begin{lemma}\label{lem:stability-rules}
  Let $r = (u,\sequence{v})$ be a rule and $\theta$
  be a substitution.
  \begin{itemize}
    \item For all positions $p$ and all terms $s,t$,
    we have: $s \ra_{(r,p)} t$ implies
    $s\theta \ra_{(r,p)} t\theta$.
    \item $\vars(v) \subseteq \vars(u)$ implies
    $\sequence{u\theta} \lpa_{(r,\sequence{1})}
    \sequence{v\theta}$.
  \end{itemize}
\end{lemma}
\begin{proof}
  Let $p$ be a position and $s,t$ be some terms such that
  $s \ra_{(r,p)} t$. Then, by Definition~\ref{def:trs-lp-rew-rel},
  we have $p \in \pos(s)$, $\subterm{s}{p} = u\sigma$ and
  $t = s[v\sigma]_p$ for some substitution $\sigma$.
  We have $p \in \pos(s\theta)$ and
  $\subterm{s\theta}{p} = \subterm{s}{p}\theta = u\sigma\theta$.
  So, by Definition~\ref{def:trs-lp-rew-rel},
  $s\theta \ra_{(r,p)} s\theta[v\sigma\theta]_p$
  where
  $s\theta[v\sigma\theta]_p = s[v\sigma]_p\theta = t\theta$.
  Consequently, we have $s\theta \ra_{(r,p)} t\theta$.

  Now, suppose that $\vars(v) \subseteq \vars(u)$ and
  let us prove that
  $\sequence{u\theta} \lpa_{(r,\sequence{1})} \sequence{v\theta}$.
  Let $(u\gamma,\sequence{v\gamma})$ be a variant of $r$
  variable disjoint with $u\theta$, for some variable
  renaming $\gamma$. Let
  $\eta=\{x\gamma \mapsto x\theta \mid x\in\vars(u),\
  x\gamma\neq x\theta\}$.
  \begin{itemize}
    \item First, we prove that $\eta$ is a substitution.
    Let $(x \mapsto s)$ and $(y \mapsto t)$ be some bindings
    in $\eta$. By definition of $\eta$, we have
    $(x \mapsto s) = (x'\gamma \mapsto x'\theta)$ and
    $(y \mapsto t) = (y'\gamma \mapsto y'\theta)$ for some
    variables $x'$ and $y'$
    in $\vars(u)$. As $\gamma$ is a variable renaming, it is a
    bijection on $X$, so if $x=y$ then $x'=y'$ and hence
    $(x \mapsto s) = (y \mapsto t)$.
    Consequently, for any bindings $(x \mapsto s)$ and
    $(y \mapsto t)$ in $\eta$, $(x \mapsto s) \neq (y \mapsto t)$
    implies $x\neq y$. Moreover, by definition of $\eta$, for
    any $(x \mapsto s) \in \eta$ we have $x \neq s$.
    Therefore, $\eta$ is a substitution.
    \item Then, we prove that $\eta$ is a unifier of $u\gamma$
    and $u\theta$.
    \begin{itemize}
      \item Let $x\in\vars(u)$. Then, by definition of $\eta$,
      $x\gamma\eta=x\theta$. So, $u\gamma\eta = u\theta$.
      \item Let $y\in\dom(\eta)$. Then, $y=x\gamma$ for some
      $x\in\vars(u)$. Hence, $y\in\vars(u\gamma)$. As $u\gamma$
      is variable disjoint with $u\theta$, we have
      $y \not\in \vars(u\theta)$. Therefore, we have
      $\dom(\eta)\cap\vars(u\theta)=\emptyset$, so
      $u\theta\eta=u\theta$.
    \end{itemize}
    Consequently, we have $u\gamma\eta = u\theta\eta$.
    \item Finally, we prove that $\eta$ is more general than
    any other unifier of $u\gamma$ and $u\theta$. Let $\sigma$
    be a unifier of $u\gamma$ and $u\theta$. Then, for all
    $x\in\vars(u)$, we have $x\gamma\sigma=x\theta\sigma$.
    For all variables $y$ that do not occur in $\dom(\eta)$,
    we have $y\eta\sigma = y\sigma$ and, for all $y \in \dom(\eta)$,
    we have $y = x\gamma$ for some $x \in \vars(u)$, so
    $y\eta\sigma = x\gamma\eta\sigma = x\theta\sigma =
    x\gamma\sigma = y\sigma$.
    Hence, $\eta\sigma=\sigma$, \ie $\eta$ is more
    general than $\sigma$.
  \end{itemize}
  Therefore, $\eta=\mgu(u\gamma,u\theta)$.
  So, by Definition~\ref{def:trs-lp-rew-rel}, we have
  $\sequence{u\theta} \lpa_{(r,\sequence{1})}
  \sequence{v\gamma\eta}$.
  Note that for all $x\in\vars(v)$, we have $x\in\vars(u)$, so
  $x\gamma\eta=x\theta$ by definition of $\eta$. Hence,
  $v\gamma\eta = v\theta$. Finally, we have
  $\sequence{u\theta} \lpa_{(r,\sequence{1})}
  \sequence{v\theta}$.
\end{proof}

\begin{example}\label{ex:lp-not-closed-subs}
  Consider the rule
  $r= (u,\sequence{v}) =
  (\fsym(x,\one),\sequence{\fsym(\one,x)})$
  and the substitution
  $\theta = \{x \mapsto \zero, y\mapsto\zero\}$.
  We have $\sequence{u\theta} = \sequence{\fsym(\zero,\one)}$,
  $\sequence{v\theta} = \sequence{\fsym(\one,\zero)}$
  and $\sequence{u\theta} \lpa_{(r,\sequence{1})}
  \sequence{v\theta}$.
  For $s = \fsym(\zero,y)$, we also have 
  $\sequence{s} \lpa_{(r,\sequence{1})}
  \sequence{\fsym(\one,\zero)}$,
  but there is no rewriting of $\sequence{s\theta}$ with
  $r$ because $s\theta = \fsym(\zero,\zero)$ 
  does not unify with any variant of $u$.
\end{example}

The next compatibility results follow from 
Definitions~\ref{def:bin-rel}
and~\ref{def:trs-lp-rew-rel} and from
Lemma~\ref{lem:stability-rules}.
\begin{lemma}\label{lem:ins-mg-compatible}
  For all rules $r$, $\ra_r$ and $\phiinst$ are
  compatible, and so are $\lpa_r$ and $\phimg$.
\end{lemma}
\begin{proof}
  Let $r$ be a rule.
  \begin{itemize}
    \item Let $s,t,s'$ be terms such that
    $s' \in \phiinst(s)$ and $s \ra_r t$.
    Then, by Definition~\ref{def:bin-rel}, we
    have $s' = c[s\sigma]$ for some context $c$
    and  some substitution $\sigma$.
    Moreover, by Definition~\ref{def:trs-lp-rew-rel},
    $s \ra_{(r,p)} t$ for some $p \in \seqset{\nat}$
    and $r$ has the form $(u,\sequence{v})$.
    So, by Lemma~\ref{lem:stability-rules}, we have
    $s\sigma \ra_{(r,p)} t\sigma$.
    Let $p'$ be the position of an occurrence of
    $\square$ in $c$. Then, 
    $c[s\sigma] \ra_{(r,p'p)} c[t\sigma]$ where
    $c[s\sigma] = s'$ and $c[t\sigma] \in \phiinst(t)$.
    So, we have proved that $s' \ra_r t'$ 
    for some $t' \in \phiinst(t)$, \ie that
    $\ra_r$ and $\phiinst$ are compatible.
    
    \item Let $\seqset{s},\seqset{t},\seqset{s'}$ be goals
    such that $\seqset{s'} \in \phimg(\seqset{s})$
    and $\seqset{s} \lpa_r \seqset{t}$.
    Then, by Definition~\ref{def:bin-rel}, we have
    $\seqset{s'} = \seqset{c}[\seqset{a}]$
    for some goal-context $\seqset{c}$
    and some goal $\seqset{a}$ that is more general
    than $\seqset{s}$.
    Moreover, by Definition~\ref{def:trs-lp-rew-rel},
    we have $\seqset{s} \lpa_{(r,p)} \seqset{t}$
    for some $p \in \seqset{\nat}$.
    Let $r'$ be a variant of $r$ that is 
    variable disjoint with $\seqset{s'}$.
    Then, $r'$ is also variable disjoint with $\seqset{a}$.
    So, by the One Step Lifting Lemma~3.21 at page~59
    of~\cite{apt97}, for some goal $\seqset{b}$ that 
    is more general than $\seqset{t}$, we have
    $\seqset{a} \lpa_{(r,p)}^{\theta} \seqset{b}$,
    where $r'$ and $\theta$ are respectively the variant
    of $r$ and the unifier used. As $r'$ is also variable
    disjoint with $\seqset{c}$, we have 
    $\seqset{c}[\seqset{a}] \lpa_{(r,p'+p)}^{\theta}
    (\seqset{c}\theta)[\seqset{b}]$ where $r'$ is the variant
    of $r$ used and $p'$ is the position of $\square$
    in $\seqset{c}$. We note that
    $\seqset{c}[\seqset{a}] = \seqset{s'}$ and
    $(\seqset{c}\theta)[\seqset{b}] \in \phimg(\seqset{t})$.
    So, we have proved that
    $\seqset{s'} \lpa_r \seqset{t'}$ for some
    $\seqset{t'} \in \phimg(\seqset{t})$, \ie that
    $\lpa_r$ and $\phimg$ are compatible.    
  \end{itemize}
\end{proof}

\begin{example}
  In Example~\ref{ex:trs-intro},
  $s = \gsym\big(\hsym(\fsym(x),\zero),x\big)$ 
  and $s \ra_{(r_3,p)} \gsym(\fsym^3(x),x)$
  where $p = \sequence{1}$.
  Let $\sigma = \{x\mapsto\zero\}$. Then,
  $s\sigma = \gsym\big(\hsym(\fsym(\zero),\zero),\zero\big)$
  and $s\sigma \ra_{(r_3,p)} \gsym(\fsym^3(\zero),\zero)$,
  \ie $s\sigma \ra_{(r_3,p)} \gsym(\fsym^3(x),x)\sigma$,
  where $s\sigma \in \phiinst(s)$ and
  $\gsym(\fsym^3(x),x)\sigma \in \phiinst(\gsym(\fsym^3(x),x))$.
\end{example}

\section{Abstract Reduction Systems}\label{sect:ars}
%
The following notion (see, \eg Chapter~2 of~\cite{baaderN98}
or Chapter~1 of~\cite{terese03}) generalises the semantics
of term rewriting and logic programming presented above.
\begin{definition}
  An \emph{abstract reduction system (ARS)} is a pair $\ars$
  consisting of a set $A$ and a \emph{rewrite relation}
  $\rra_{\Pi}$, which is the union of binary relations on $A$
  indexed by a set $\Pi$, \ie
  $\rra_{\Pi} = \bigcup \{\rra_{\pi} \mid \pi \in \Pi\}$. 
\end{definition}
We have $A = \termset$ in term rewriting and
$A = \seqset{\termset}$ in logic programming;
moreover, $\Pi$ is a program in both cases.

We formalise non-termination as the existence of an infinite
chain in an ARS.
\begin{definition}
  A \emph{chain} in an ARS $\cA = \ars$ is a (possibly infinite)
  $\rra_{\Pi}$-chain.
\end{definition}

\subsection{Closure Under Substitutions}
\label{sect:closure-subs}
In Section~\ref{sect:rec-pairs}, we prove that the
existence of finite chains $u_1 \rra^+_{\Pi} v_1$ and
$u_2 \rra^+_{\Pi} v_2$ of a special form implies that
of an infinite chain that involves instances of $u_1$,
$v_1$, $u_2$ and $v_2$ (see Definition~\ref{def:rec-pair}
and Corollary~\ref{coro:w1-w2}). The proof relies on
the following property of ARSs. In the rest of this paper,
for all ARSs $\ars$ and all
$w = \sequence{\pi_1,\dots,\pi_n}$ in
$\seqset{\Pi}$, we let
$\rra_w = (\rra_{\pi_1} \circ \cdots \circ \rra_{\pi_n})$,
where $\rra_{\epsilon}$ is the identity relation.

\begin{definition}\label{def:closed-subs}
  Let $\cA = \ars$ be an ARS where
  $A \subseteq \termset \cup \seqset{\termset}$.
  We say that $\cA$ is
  \emph{closed under substitutions} if, for all
  $s,t \in A$, all $w \in \seqset{\Pi}$ and all
  substitutions $\theta$, $s \rra_w t$ implies
  $s\theta \rra_w t\theta$.
\end{definition}

The following result is a consequence
of Lemma~\ref{lem:stability-rules}.
\begin{lemma}\label{lem:closed-subs-trs}
  For all programs $P$, $(\termset, \ra_P)$
  is closed under substitutions.
\end{lemma}
\begin{proof}
  Let $P$ be a program,
  $s,t$ be terms, $w \in \seqset{P}$ and $\theta$
  be a substitution. Suppose that $s \ra_w t$. 
  We prove by induction on $\length{w}$ that
  $s\theta \ra_w t\theta$.
  \begin{itemize}
    \item (Base: $\length{w} = 0$) Here, $\ra_w$ is the
    identity relation. As $s \ra_w t$, we have $s = t$,
    hence $s\theta = t\theta$, and so $s\theta \ra_w t\theta$.
    \item (Induction) Suppose that $\length{w} = n + 1$
    for some $n \in \nat$. Suppose also that for all terms
    $s',t'$, all $w' \in \seqset{P}$ with $\length{w'} = n$
    and all substitutions $\sigma$, $s' \ra_{w'} t'$ implies
    $s'\sigma \ra_{w'} t'\sigma$.
    As $\length{w} = n + 1$, we have $w = rw'$ for some
    $r \in P$ and some $w' \in \seqset{P}$ with
    $\length{w'} = n$. Therefore, $s \ra_r s' \ra_{w'} t$
    for some term $s'$. By definition of $\ra_r$, we have
    $r = (u, \sequence{v})$, so, by Lemma~\ref{lem:stability-rules},  
    $s\theta \ra_r s'\theta$. As $\length{w'} = n$,
    by induction hypothesis we have
    $s'\theta \ra_{w'} t\theta$. Hence, 
    $s\theta \ra_r s'\theta \ra_{w'} t\theta$,
    \ie $s\theta \ra_w t\theta$.
  \end{itemize}
\end{proof}

However, for all programs $P$,
$(\seqset{\termset}, \lpa_P)$ is not closed under
substitutions, see Example~\ref{ex:lp-not-closed-subs}.
Hence, based on Lemma~\ref{lem:stability-rules}, we
introduce the following restricted form  of logic
programming, where one only rewrites singleton goals
using rules, the right-hand side of which is a 
singleton goal.

\begin{definition}\label{def:restricted-lp}
  For all programs $P$, we let 
  $\hra_P = \bigcup \left\{\hra_r \;\middle\vert\; r\in P \right\}$
  where, for all $r \in P$,
  $\hra_r = \left\{(u\theta, v\theta) \in \termset^2
  \;\middle\vert\; r = (u, \sequence{v}),\
  \vars(v) \subseteq \vars(u),\
  \theta \in \subsset\right\}$.
\end{definition}

We note that $\hra_r \subseteq \ra_{r,\epsilon}$.
Moreover, $\hra_r = \emptyset$ if $r = (u,\seqset{v})$
with $\vars(\seqset{v}) \not\subseteq \vars(u)$ or
$\length{\seqset{v}} \neq 1$.
Now, we have a counterpart of Lemma~\ref{lem:closed-subs-trs}
in logic programming:
\begin{lemma}\label{lem:closed-subs-lp}
  For all programs $P$, $(\termset, \hra_P)$
  is closed under substitutions.
\end{lemma}
\begin{proof}
  Let $P$ be a program,
  $s,t$ be terms, $w \in \seqset{P}$ and $\theta$
  be a substitution. Suppose that $s \hra_w t$. 
  We prove by induction on $\length{w}$ that 
  $s\theta \hra_w t\theta$.
  \begin{itemize}
    \item (Base: $\length{w} = 0$) Here, $\hra_w$ is the
    identity relation. As $s \hra_w t$, we have $s = t$,
    hence $s\theta = t\theta$, and so $s\theta \hra_w t\theta$.
    \item (Induction) Suppose that $\length{w} = n + 1$
    for some $n \in \nat$. Suppose also that for all terms
    $s',t'$, all $w' \in \seqset{P}$ with $\length{w'} = n$
    and all substitutions $\sigma$, $s' \hra_{w'} t'$ implies
    $s'\sigma \hra_{w'} t'\sigma$.
    As $\length{w} = n + 1$, we have $w = rw'$ for some
    $r \in P$ and some $w' \in \seqset{P}$ with
    $\length{w'} = n$. Therefore, $s \hra_r s' \hra_{w'} t$
    for some term $s'$. By Definition~\ref{def:restricted-lp},
    we have $r = (u, \sequence{v})$, $\vars(v) \subseteq \vars(u)$,
    $s = u\sigma$ and $s' = v\sigma$ for some substitution
    $\sigma$; moreover, $u\sigma\theta \hra_r v\sigma\theta$;
    so, $s\theta \hra_r s'\theta$. As $\length{w'} = n$,
    by induction hypothesis we have
    $s'\theta \hra_{w'} t\theta$. Hence, 
    $s\theta \hra_r s'\theta \hra_{w'} t\theta$,
    \ie $s\theta \hra_w t\theta$.
  \end{itemize}
\end{proof}

It follows from the next result that the existence
of an infinite $\hra_P$-chain implies that of
an infinite $\lpa_P$-chain, \ie non-termination
in the restricted form of logic programming
implies non-termination in full logic programming.
\begin{lemma}\label{lem:chain-restricted-lp}
  For all programs $P$, terms $s,t$ and
  $w \in \seqset{P}$, $s \hra_w t$ implies
  $\sequence{s} \lpa_w \sequence{t}$.
\end{lemma}
\begin{proof}
  Let $P$ be a program, $s,t$ be terms and
  $w \in \seqset{P}$. Suppose that $s \hra_w t$.
  We prove by induction on $\length{w}$ that
  $\sequence{s} \lpa_w \sequence{t}$.
  \begin{itemize}
    \item (Base: $\length{w} = 0$) Here, $\hra_w$ and
    $\lpa_w$ are the identity relation. As $s \hra_w t$,
    we have $s = t$, hence $\sequence{s} = \sequence{t}$,
    and so $\sequence{s} \lpa_w \sequence{t}$.
    \item (Induction) Suppose that $\length{w} = n + 1$
    for some $n \in \nat$. Suppose also that for all terms
    $s',t'$ and all $w' \in \seqset{P}$ with $\length{w'} = n$,
    $s' \hra_{w'} t'$ implies
    $\sequence{s'} \lpa_{w'} \sequence{t'}$.
    As $\length{w} = n + 1$, we have $w = rw'$ for some
    $r \in P$ and some $w' \in \seqset{P}$ with
    $\length{w'} = n$. Therefore, $s \hra_r s' \hra_{w'} t$
    for some term $s'$.
    By definition of $\hra_r$, we have
    $r = (u, \sequence{v})$, $\vars(v) \subseteq \vars(u)$,
    $s = u\theta$ and $s' = v\theta$ for some substitution
    $\theta$. So, by Lemma~\ref{lem:stability-rules}, we have 
    $\sequence{s} \lpa_r \sequence{s'}$. As $\length{w'} = n$,
    by induction hypothesis we have
    $\sequence{s'} \lpa_{w'} \sequence{t}$. Hence, 
    $\sequence{s} \lpa_r \sequence{s'} \lpa_{w'} \sequence{t}$,
    \ie $\sequence{s} \lpa_w \sequence{t}$.
  \end{itemize}
\end{proof}

\section{Loops}\label{sect:loops}
In Section~\ref{sect:intro-loops}, we have provided an
informal description of loops. Now, we propose a formal
definition in an abstract setting.

\begin{definition}\label{def:loop}
  Let $\cA = \ars$ be an ARS, $w \in \seqset{\Pi}$ and
  $\phi$ be a binary relation on $A$ which is compatible 
  with $\rra_{\Pi}$.
  A $(w,\phi)$-\emph{loop} in $\cA$ is a pair
  $(a,a') \in A^2$ such that $a \rra_w a'$ and
  $a' \in \phi(a)$.
\end{definition}

By definition of compatibility
(Definition~\ref{def:compatibility}), the existence of a
loop immediately leads to that of an infinite chain.
\begin{lemma}\label{lem:compatibility-infinite}
  Suppose that there is a $(w,\phi)$-loop $(a,a')$ in an
  ARS $\ars$. Then, there is an infinite $\rra_w$-chain
  that starts with $a \rra_w a'$.
\end{lemma}
\begin{proof}
  Let $a_0 = a$ and $a_1 = a'$.
  As $(a,a')$ is a $(w,\phi)$-loop, we have
  $a_0 \rra_w a_1$ and $a_1 \in \phi(a_0)$.
  Moreover, $\rra_{\Pi}$ and $\phi$ are compatible.
  Therefore, by the definition of $\rra_w$ and
  Lemma~\ref{lem:compatibility-composition},
  $\rra_w$ and $\phi$ are compatible. So, we have
  $a_1 \rra_w a_2$ for some $a_2 \in \phi(a_1)$. Again by
  compatibility of $\rra_w$ and $\phi$, we have $a_2 \rra_w a_3$
  for some $a_3 \in \phi(a_2)$, \emph{etc}.
\end{proof}

We note that the infinite chain $a_0 \rra_w a_1 \rra_w \cdots$
corresponding to a $(w,\phi)$-loop is such that for
all $n \in \nat$, $a_0 \rra^n_w a_n$ and $a_n \in \phi^n(a_0)$,
\ie elements that are somehow ``similar'' to $a_0$ \wrt{}
$\phi$ (\ie $a_n \in \phi^n(a_0)$) are periodically ``reached'' 
via $\rra_w$. The name \emph{loop} stems from this observation.

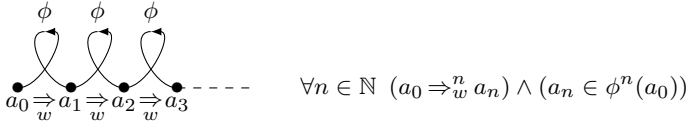
\begin{center}
  \begin{tikzpicture}
    \draw (0,0)   node {$\bullet$} node[below] {\small{}$a_0$} ;
    \draw (0.7,0) node {$\bullet$} node[below] {\small{}$a_1$} ;
    \draw (1.4,0) node {$\bullet$} node[below] {\small{}$a_2$} ;
    \draw (2.1,0) node {$\bullet$} node[below] {\small{}$a_3$} ;

    \draw (0.35,0) node[below] {\small{}$\rra$};
    \draw (1.05,0) node[below] {\small{}$\rra$};
    \draw (1.75,0) node[below] {\small{}$\rra$};

    \draw (0.35,-0.2) node[below] {\footnotesize{}$w$};
    \draw (1.05,-0.2) node[below] {\footnotesize{}$w$};
    \draw (1.75,-0.2) node[below] {\footnotesize{}$w$};

    \draw[->,>=latex](0,0) to[out=0, in=0] (0.35,0.75) ;
    \draw (0.35,0.75) to[out=180, in=180] (0.7,0) ; 
    \draw[->,>=latex](0.7,0) to[out=0, in=0] (1.05,0.75) ;
    \draw (1.05,0.75) to[out=180, in=180] (1.4,0) ; 
    \draw[->,>=latex](1.4,0) to[out=0, in=0] (1.75,0.75) ;
    \draw (1.75,0.75) to[out=180, in=180] (2.1,0) ; 
    
    \draw (0.35,0.75) node [above] {\small{}$\phi$} ;
    \draw (1.05,0.75) node [above] {\small{}$\phi$} ;
    \draw (1.75,0.75) node [above] {\small{}$\phi$} ;

    \draw[dashed] (2.1,0) -- (3.1,0) ;

    \draw (3.6,0) node [right]
    {\small{}$\forall n\in\nat\ \left(a_0 \rra^n_w a_n\right)
    \land \left(a_n \in \phi^n(a_0)\right)$};
  \end{tikzpicture}
\end{center}

Moreover, the chain $a_0 \rra_w a_1 \rra_w \cdots$ only
relies on a single sequence $w$ of elements of $\Pi$.
In the next section, we will consider more elaborated
chains based on two sequences. 

\begin{example}\label{ex:looping-trs}
  In Example~\ref{ex:trs-intro-2}, we have
  \[\underbrace{\fsym(x)}_{a_0} \ra_{r_1}
  \gsym(\hsym(x,\one),x) \ra_{r_2}
  \gsym(\hsym(x,\zero),x) \ra_{r_3}
  \underbrace{\gsym(\fsym^2(x),x)}_{a_3}\]
  \ie $a_0 \ra_w a_3$ for $w = \sequence{r_1, r_2, r_3}$.
  Moreover, $a_3 \in \phiinst(a_0)$ and $\ra_P$ and
  $\phiinst$ are compatible, where $P = \{r_1,r_2,r_3\}$
  (see Lemma~\ref{lem:ins-mg-compatible}).
  So, $(a_0, a_3)$ is a $(w,\phiinst)$-loop in
  $(\termset,\ra_P)$. Therefore, by
  Lemma~\ref{lem:compatibility-infinite}, there
  is an infinite $\ra_w$-chain that starts with
  $a_0 \ra_w a_3$. Indeed, we have
  \[\underbrace{\fsym(x)}_{a_0} \ra_w
  \underbrace{\gsym(\fsym^2(x),x)}_{a_3} \ra_w
  \underbrace{\gsym\big(\gsym(\fsym^3(x),\fsym(x)),x\big)}_{a_6} \ra_w
  \cdots\]
  where $a_3 \in \phiinst(a_0)$, $a_6 \in \phiinst(a_3)$, \dots{}
  We note that in this chain, the rules are applied at positions
  that vary gradually (\eg $r_1$ is applied to subterms of the form
  $\fsym(\cdots)$ that occur at deeper and deeper positions).
\end{example}

\begin{example}\label{ex:looping-lp}
  In Example~\ref{ex:lp-intro-2}, we have
  \[\underbrace{\sequence{\psym(\fsym(x,\zero))}}_{\seqset{a_0}}
  \lpa_r
  \underbrace{\sequence{\psym(x),\qsym(x)}}_{\seqset{a_1}}\]
  \ie $\seqset{a_0} \ra_w \seqset{a_1}$ for
  $w = \sequence{r}$. Moreover,
  $\seqset{a_1} \in \phimg(\seqset{a_0})$ and
  $\lpa_P$ and $\phimg$ are compatible, where $P = \{r\}$
  (see Lemma~\ref{lem:ins-mg-compatible}).
  So, $(\seqset{a_0},\seqset{a_1})$
  is a $(w,\phimg)$-loop in $(\seqset{\termset},\lpa_P)$.
  Therefore, by Lemma~\ref{lem:compatibility-infinite},
  there is an infinite $\lpa_w$-chain that starts
  with $\seqset{a_0} \lpa_w \seqset{a_1}$. Indeed, we have
  \[\underbrace{\sequence{\psym(\fsym(x,\zero))}}_{\seqset{a_0}} \lpa_w
  \underbrace{\sequence{\psym(x),\qsym(x)}}_{\seqset{a_1}} \lpa_w
  \underbrace{\sequence{\psym(x_1),\qsym(x_1),
  \psym(\fsym(x_1,\zero))}}_{\seqset{a_2}} \lpa_w
  \cdots\]
  where $\seqset{a_1} \in \phimg(\seqset{a_0})$,
  $\seqset{a_2} \in \phimg(\seqset{a_1})$, \dots
\end{example}

\section{Binary Chains}\label{sect:binary-chains}
In Section~\ref{sect:loops} we have considered infinite
chains that rely on a single sequence $w$ of elements
of $\Pi$. A natural way to extend this class
is to consider infinite chains based on two sequences
$w_1$ and $w_2$ as in Definition~\ref{def:binary-chain}
below.
In principle, one could also define many other forms
of similar chains. In this paper, we concentrate on this
form only as it covers interesting examples found in
the literature or in~\cite{tpdb} (\eg
Examples~\ref{ex:inner-loop},
\ref{ex:binary-chain-tpdb},
\ref{ex:binary-chain-zantemaG96} below)
and, moreover, we are able to provide an automatable
approach for the detection of a special case
(see Section~\ref{sect:rec-pairs}); this case
allows one for instance to encode the semantics of
vector addition systems~\cite{karpM69}.

\begin{definition}\label{def:binary-chain}
    A \emph{binary chain} in an ARS $\ars$ is an infinite
    $(\rra^*_{w_1} \circ \rra_{w_2})$-chain for some
    $w_1,w_2 \in \seqset{\Pi}$.
\end{definition}

Of course, as $\rra_w = (\rra^0_w \circ \rra_w)$,
the infinite chain $a_0 \rra_w a_1 \rra_w \cdots$
corresponding to a $(w,\phi)$-loop is also binary.
Moreover, any infinite chain of the form
$a_0 \mathop{(\rra^*_{w_1} \circ \rra^*_{w_2})} a_1
\mathop{(\rra^*_{w_1} \circ \rra^*_{w_2})} \cdots$
is binary, because any sequence 
$a_i \rra_{w_2} a'_i \rra_{w_2} a''_i$ has the form
$a_i \rra^0_{w_1} a_i \rra_{w_2} a'_i \rra^0_{w_1}
a'_i \rra_{w_2} a''_i$.

\begin{example}[\cite{dershowitz87,geserZ99,wangS06}]
    \label{ex:inner-loop}
    Consider the rules
    \[r_1 = \big(\bsym(\csym), \sequence{\dsym(\csym)}\big) \qquad
    r_2 = \big(\bsym(\dsym(x)), \sequence{\dsym(\bsym(x))}\big) \qquad
    r_3 = \big(\asym(\dsym(x)), \sequence{\asym(\bsym^2(x))}\big)\]
    For the sake of readability, let us omit parentheses
    and write for instance $\symb{adc}$ instead of
    $\asym(\dsym(\csym))$. Let $w_1 = \sequence{r_2}$ and
    $w_2 = \sequence{r_3,r_1}$. We have the infinite
    $(\ra^*_{w_1} \circ \ra_{w_2})$-chain
    \begin{align*}
        \symb{adc}
        & \ra^0_{r_2} \symb{adc} 
        \ra_{r_3} \symb{ab}^2\csym
        \ra_{r_1} \symb{abdc} \\
        & \ra^1_{r_2} \symb{adbc}
        \ra_{r_3} \symb{ab}^3\csym
        \ra_{r_1} \symb{ab}^2\symb{dc} \\
        & \ra^2_{r_2} \symb{adb}^2\csym
        \ra_{r_3} \symb{ab}^4\csym
        \ra_{r_1} \symb{ab}^3\symb{dc} \\
        & \ra^3_{r_2} \cdots
    \end{align*}
    We note that in this chain, the rule $r_3$
    is always applied at the root position but that 
    $r_1$ and $r_2$ are applied at positions that vary.
\end{example}

\begin{example}[\texttt{TRS\_Standard/Zantema\_15/ex11.xml} in~\cite{tpdb}]
  \label{ex:binary-chain-tpdb}
  Consider the rules
  \[r_1 = \big(\fsym(x, \ssym(y)), \sequence{\fsym(\ssym(x), y)}\big)
  \qquad
  r_2 = \big(\fsym(x, \zero), \sequence{\fsym(\ssym(\zero), x)}\big)\]
  Let $w_1 = \sequence{r_1}$ and $w_2 = \sequence{r_2}$.
  We have the infinite $(\hra^*_{w_1} \circ \hra_{w_2})$-chain
  \begin{align*}
    \fsym(\ssym(\zero), \zero)
     & \hra^0_{r_1}
     \fsym(\ssym(\zero), \zero)
     \hra_{r_2}
     \fsym(\ssym(\zero), \ssym(\zero)) \\
     & \hra^1_{r_1}
     \fsym(\ssym^2(\zero), \zero)
     \hra_{r_2}
     \fsym(\ssym(\zero), \ssym^2(\zero)) \\
     & \hra^2_{r_1}
     \fsym(\ssym^3(\zero), \zero)
     \hra_{r_2}
     \fsym(\ssym(\zero), \ssym^3(\zero)) \\
     & \hra^3_{r_1} \cdots
  \end{align*}
  For instance, $u_1 = \fsym(\ssym(\zero), \ssym(\zero))$
  and $u_2 = \fsym(\ssym^2(\zero), \zero)$ correspond to
  the integer vectors $v_1 = (1,1)$ and $v_2 = (2,0)$,
  respectively, and the chain $u_1 \hra_{r_1} u_2$ models
  the componentwise addition of the vector $(1,-1)$ to
  $u_1$.
  The programs \verb+ex12.xml+ and \verb+ex14.xml+
  in the same directory of~\cite{tpdb} are similar.
\end{example}

\begin{example}[\cite{zantemaG96}]
  \label{ex:binary-chain-zantemaG96}
  Consider the rules
  \[r_1 = \big(\fsym(\csym, \asym(x), y), \sequence{\fsym(\csym, x, \asym(y))}\big)
  \qquad
  r_2 = \big(\fsym(\csym, \asym(x), y), \sequence{\fsym(x, y, \asym^2(\csym))}\big)\]
  Let $w_1 = \sequence{r_1}$ and $w_2 = \sequence{r_2}$.
  We have the infinite $(\hra^*_{w_1} \circ \hra_{w_2})$-chain
  \begin{align*}
    \fsym(\csym, \asym(\csym), \asym^2(\csym))
    & \hra^0_{r_1}
    \fsym(\csym, \asym(\csym), \asym^2(\csym))
    \hra_{r_2}
    \fsym(\csym, \asym^2(\csym), \asym^2(\csym)) \\
    & \hra^1_{r_1}
    \fsym(\csym, \asym(\csym), \asym^3(\csym))
    \hra_{r_2}
    \fsym(\csym, \asym^3(\csym), \asym^2(\csym)) \\
    & \hra^2_{r_1}
    \fsym(\csym, \asym(\csym), \asym^4(\csym))
    \hra_{r_2}
    \fsym(\csym, \asym^4(\csym), \asym^2(\csym)) \\
    & \hra^3_{r_1} \cdots
  \end{align*}
\end{example}

\subsection{Recurrent Pairs}
\label{sect:rec-pairs}
Now, we present a new criterion for the detection of binary
chains. It is based on two specific chains
$u_1 \rra^+_{\Pi} v_1$ and $u_2 \rra^+_{\Pi} v_2$ such that
a context is removed from $u_1$ to $v_1$ while it is added
again from $u_2$ to $v_2$. This is formalised as follows.
In the next definition, we consider a new hole symbol
$\square'$ that does not occur in
$\Sigma \cup X \cup \{\square\}$ and we let $c_1$ be a
context with at least one occurrence of $\square$ and
$\square'$; moreover, for all terms $t,t'$, we let $c_1[t,t']$
be the term obtained from $c_1$ by replacing the occurrences of
$\square$ (resp. $\square'$) by $t$ (resp. $t'$). On the other
hand, we let $c_2$ be a context with occurrences of $\square$
only (as in Definition~\ref{def:term-context}).

\begin{definition}\label{def:rec-pair}
  Let $\cA = (\termset,\rra_{\Pi})$ be an ARS
  closed under substitutions.
  A \emph{recurrent pair in $\cA$} is a pair
  $(u_1 \rra_{w_1} v_1, u_2 \rra_{w_2} v_2)$
  of finite chains in $\cA$ such that 
  \begin{itemize}
    \item $u_1 = c_1[x, c_2[y]]$, $v_1 = c_1[c_2^{n_1}[x], y]$,
    $u_2 = c_1[x, c_2^{n_2}[s]]$ and $v_2 = c_1[c_2^{n_3}[t], c_2^{n_4}[x]]$
    \item $x \neq y$ and $\{x,y\} \cap \vars(c_1) = \emptyset$
    \item $\vars(c_2) = \vars(s) = \emptyset$
    \item $t \in \{x,s\}$
    \item $n_4 \geq n_2$
  \end{itemize}
\end{definition}

The ARS $\left(\termset,\ra_{\{r_1,r_2,r_3\}}\right)$ of
Example~\ref{ex:inner-loop} is not covered by this definition
because it only involves function symbols of arity 0 or 1
(hence, one cannot find a context $c_1$ with at least one
occurrence of $\square$ and $\square'$).

\begin{example}\label{ex:trs-binary-intro-2}
  In Example~\ref{ex:trs-binary-intro}, we have the chains
  $\fsym(x, c[y], x) \ra_{w_1} \fsym(c[x], y, c[x])$
  and
  $\fsym(x, \zero, x) \ra_{w_2} \fsym(c[x], c[x], c[x])$.
  As $\big(\termset,\ra_{\{r_1,r_2,r_3,r_4\}}\big)$ is
  closed under substitutions (Lemma~\ref{lem:closed-subs-trs}),
  these chains form a recurrent pair, with
  $c_1 = \fsym(\square,\square',\square)$,
  $c_2 = c = \gsym(\square,\zero,\square)$,
  $(n_1,n_2,n_3,n_4) = (1,0,1,1)$, $s = \zero$ and $t = x$. 
\end{example}

\begin{example}\label{ex:binary-chain-tpdb-2}
  In Example~\ref{ex:binary-chain-tpdb}, we have the chains
  $\fsym(x, \ssym(y)) \hra_{w_1} \fsym(\ssym(x), y)$
  and
  $\fsym(x, \zero) \hra_{w_2} \fsym(\ssym(\zero), x)$.
  As $\big(\termset,\hra_{\{r_1,r_2\}}\big)$ is closed
  under substitutions (Lemma~\ref{lem:closed-subs-lp}),
  these chains form a recurrent pair, with
  $c_1 = \fsym(\square,\square')$,
  $c_2 = \ssym(\square)$, $(n_1,n_2,n_3,n_4) = (1,0,1,0)$
  and $s = t = \zero$.
\end{example}

\begin{example}\label{ex:binary-chain-zantemaG96-2}
  In Example~\ref{ex:binary-chain-zantemaG96}, we have
  the chains 
  $\fsym(\csym, \asym(x), y) \hra_{w_1}
  \fsym(\csym, x, \asym(y))$
  and
  $\fsym(\csym, \asym(\csym), y) \hra_{w_2}
  \fsym(\csym, y, \asym^2(\csym))$.
  As $\big(\termset,\hra_{\{r_1,r_2\}}\big)$ is closed
  under substitutions (Lemma~\ref{lem:closed-subs-lp}),
  these chains form a recurrent pair, with
  $c_1 = \fsym(\csym,\square',\square)$,
  $c_2 = \asym(\square)$, $(n_1,n_2,n_3,n_4) = (1,0,1,0)$
  and $s = t = \asym(\csym)$.
\end{example}

It is proved in~\cite{payet23} that the existence of a recurrent
pair of a more restricted form (\ie $\length{w_1} = \length{w_2} = 1$,
$c_1 = \fsym(\square,\square')$ and $n_2 = 0$, as in
Example~\ref{ex:binary-chain-tpdb}) leads to that of a binary
chain. The generalisation to any sequences $w_1,w_2$, any
context $c_1$ and any $n_2 \in \nat$ satisfying the constraints
above is presented below (Corollary~\ref{coro:w1-w2}).
In some special situations, the obtained binary chain actually
relies on a single sequence (\eg if $n_2 = n_3 = n_4 = 0$ then
we have  $c_1[s,s] \ra_{w_2} c_1[s,s] \ra_{w_2} \cdots$), but
this is not always the case (see the examples above).
The next statements are parametric in an ARS $\cA$ closed
under substitutions and a recurrent pair in $\cA$, with
the notations of Definition~\ref{def:rec-pair} as well as
this new one (introduced for the sake of readability):
\begin{definition}
  For all $m, n \in \nat$, we let $c_1[m, n]$ denote
  the term $c_1[c_2^m[s],c_2^n[s]]$. 
\end{definition}

Then, we have the following two results.
Lemma~\ref{lem:reductions-w1} states that $w_1$
allows one to iteratively move a tower of $c_2$'s
from the positions of $\square'$ to those of $\square$
in $c_1$. Conversely, Lemma~\ref{lem:reductions-w2}
states that $w_2$ allows one to copy a tower of $c_2$'s
from the positions of $\square$ to those of $\square'$ in
just one application of $\rra_{w_2}$.
\begin{lemma}\label{lem:reductions-w1}
  For all $m,n \in \nat$,
  $c_1[m, n+1] \rra_{w_1} c_1[m + n_1, n]$.
  Consequently, for all $m,n \in \nat$ with $n \geq n_2$, we have
  $c_1[m, n] \rra^{n-n_2}_{w_1} c_1[m + (n-n_2) \times n_1, n_2]$.
\end{lemma}
\begin{proof}
  Let $m,n \in \nat$. Then,
  $c_1[m, n+1] = c_1[c_2^m[s], c_2^{n+1}[s]] = u_1\theta$
  where $\theta = \{x \mapsto c_2^m[s], y \mapsto c_2^n[s]\}$.
  So, as $u_1 \rra_{w_1} v_1$ and $\cA$ is closed under
  substitutions, we have
  $c_1[m, n+1] \rra_{w_1} v_1\theta$ where
  $v_1\theta = c_1[c_2^{m + n_1}[s],c_2^n[s]] = c_1[m + n_1, n]$.

  Now, we prove the second part of the lemma by induction on $n$.
  \begin{itemize}
    \item (Base: $n = n_2$) Here, $\rra^{n-n_2}_{w_1}$ is the
    identity relation. Hence, for all $m \in \nat$, we have
    $c_1[m, n] \rra^{n-n_2}_{w_1} c_1[m, n]$,
    where $c_1[m, n] = c_1[m + (n-n_2) \times n_1, n_2]$.
    \item (Induction) Suppose that for some $n \geq n_2$
    we have $c_1[m, n] \rra^{n-n_2}_{w_1} c_1[m + (n-n_2) \times n_1, n_2]$
    for all $m \in \nat$. Let $m \in \nat$. By the first part
    of the lemma, $c_1[m, n+1] \rra_{w_1} c_1[m + n_1, n]$.
    Moreover, by induction hypothesis,
    $c_1[m + n_1, n] \rra^{n-n_2}_{w_1}
    c_1[(m + n_1) + (n-n_2) \times n_1, n_2]$, \ie
    $c_1[m + n_1, n] \rra^{n-n_2}_{w_1}
    c_1[m + (n + 1 - n_2) \times n_1, n_2]$.
    Consequently, finally we have
    $c_1[m, n+1] \rra^{n+1 - n_2}_{w_1}
    c_1[m + (n + 1 - n_2) \times n_1, n_2]$.
  \end{itemize}
\end{proof}

\begin{lemma}\label{lem:reductions-w2}
  For all $m \in \nat$, we have
  $c_1[m, n_2] \rra_{w_2} c_1[m' + n_3, m + n_4]$
  where $m' = 0$ if $t = s$ and $m' = m$ if $t = x$.
\end{lemma}
\begin{proof}
  Let $m \in \nat$. We have
  $c_1[m, n_2] = c_1[c_2^m[s], c_2^{n_2}[s]] = u_2\{x \mapsto c_2^m[s]\}$.
  Hence, as $u_2 \rra_{w_2} v_2$ and $\cA$ is closed under
  substitutions, we have
  $c_1[m, n_2] \rra_{w_2} v_2\{x \mapsto c_2^m[s]\}$.
  \begin{itemize}
    \item If $t = s$ then $v_2\{x \mapsto c_2^m[s]\} =
    c_1[c_2^{n_3}[s],c_2^{m+n_4}[s]] = c_1[n_3, m+n_4]$.
    \item If $t = x$ then $v_2\{x \mapsto c_2^m[s]\} =
    c_1[c_2^{m+n_3}[s],c_2^{m+n_4}[s]] = c_1[m+n_3, m+n_4]$.
  \end{itemize}
\end{proof}

By combining Lemmas~\ref{lem:reductions-w1}
and~\ref{lem:reductions-w2}, one gets:
\begin{proposition}\label{prop:w1-w2}
  For all $m,n \in \nat$ with $n \geq n_2$,
  there exist $m',n' \in \nat$ such that $n' \geq n_2$ and 
  $c_1[m, n] \mathop{(\rra^*_{w_1} \circ \rra_{w_2})}
  c_1[m', n']$.
\end{proposition}
\begin{proof}
  Let $m,n \in \nat$ with $n \geq n_2$.
  By Lemma~\ref{lem:reductions-w1}, 
  $c_1[m, n] \rra^{n-n_2}_{w_1} c_1[l, n_2]$
  where $l = m + (n-n_2) \times n_1$.
  Moreover, by Lemma~\ref{lem:reductions-w2}, 
  $c_1[l, n_2] \rra_{w_2} c_1[l' + n_3, l + n_4]$
  where $l' = 0$ if $t = s$ and $l' = l$
  if $t = x$. Therefore, for
  $m' = l' + n_3$ and $n' = l + n_4$, we
  have $c_1[m, n] \mathop{(\rra^{n-n_2}_{w_1} \circ \rra_{w_2})}
  c_1[m', n']$ and we note that $n' \geq n_2$ because,
  by Definition~\ref{def:rec-pair}, $n_4 \geq n_2$.
\end{proof}

The next result is a straighforward consequence of
Proposition~\ref{prop:w1-w2}.
\begin{corollary}\label{coro:w1-w2}
  For all $m,n \in \nat$ with $n \geq n_2$,
  $c_1[m, n]$ starts an infinite
  $\mathop{(\rra^*_{w_1} \circ \rra_{w_2})}$-chain.
\end{corollary}

\begin{example}[Example~\ref{ex:trs-binary-intro-2} continued]
  An infinite $\mathop{(\ra^*_{w_1} \circ \ra_{w_2})}$-chain that
  starts from $\fsym(c[\zero],c[\zero],c[\zero]) = c_1[1,1]$ is
  provided in Example~\ref{ex:trs-binary-intro}.
\end{example}

\begin{example}[Example~\ref{ex:binary-chain-tpdb-2}
  (resp~\ref{ex:binary-chain-zantemaG96-2}) continued]
  An infinite $\mathop{(\hra^*_{w_1} \circ \hra_{w_2})}$-chain that
  starts from $\fsym(\ssym(\zero),\zero) = c_1[1,0]$
  (resp. $\fsym(\csym, \asym(\csym), \asym^2(\csym)) = c_1[1,0]$)
  is provided in Example~\ref{ex:binary-chain-tpdb}
  (resp.~\ref{ex:binary-chain-zantemaG96}).
\end{example}

\section{Experimental Evaluation}\label{sect:implementation}

\begin{table}[h]
  \centering
  \caption{Termination~Competition~2023 -- Category \emph{TRS Standard}.}
  \label{table:trs}
  \begin{tabular}{|c||l@{\;}l|l@{\;}l|r@{\;}l|c|r@{\;}l|r@{\;}l|}
    \hline
    Participant &
    \multicolumn{2}{c|}{\textsf{NO}'s} &
    \multicolumn{2}{c|}{\% VBS} & 
    \multicolumn{2}{c|}{(E12)} &
    (P21) &
    \multicolumn{2}{c|}{(P23)} &
    \multicolumn{2}{c|}{(Z15)} \\
    \hline \hline
    \aprove & 278 & (22) & 81.8 & (6.5) & 35 & (17) & 0 & 0 & & 2 & \\
    \hline
    \autonon & 233 & (20) & 68.5 & (5.9) & 19 & (3) & 3 & 0 & & 14 & (9) \\
    \hline
    \mnm & 253 & (9) & 74.4 & (2.6) & 0 & & 0 & 0 & & 0 & \\
    \hline
    \muterm & 136 & & 40.0 & & 0 & & 0 & 0 & & 0 & \\
    \hline
    \natt   & 170 & & 50.0 & & 0 & & 0 & 0 & & 0 & \\
    \hline
    \nti    & 263 & (10) & 77.4 & (2.9) & 5 & & 3 & 10 & (10) & 3 & \\
    \hline
    \tttt   & 194 & & 57.1 & & 0 & & 0 & 0 & & 0 & \\
    \hline
  \end{tabular}
\end{table}

In~\cite{payet08,payet18,payetM06}, we have provided
details and experimental evaluations with numbers 
about the loop detection approach of \nti. Since then,
the tool has evolved, as we have fixed several bugs and
implemented the syntactic criterion of
Section~\ref{sect:rec-pairs} for detecting binary chains.
Hence, we present an updated evaluation based on the results
of \nti{} at the
Termination~Competition~2023~\cite{termcomp}.
\nti{} participated in two categories.

\begin{itemize}
  \item Category \emph{Logic Programming} (315~LPs,
  2~participants: \aprove~\cite{aprove,aproveWeb} and \nti).
  Here, \nti{} was the only tool capable of detecting
  non-termination. In total, it found 58~non-terminating LPs;
  45~of them were detected from a loop and the other~13
  from a recurrent pair.
  \item Category \emph{TRS Standard} (1523~TRSs, 7~participants:
  \aprove, \autonon~\cite{endrullisZ15}, \mnm~\cite{mnm},
  \muterm~\cite{mutermWeb}, \natt~\cite{nattWeb}, \nti,
  \tttt~\cite{tttWeb}).
  We report some results in Table~\ref{table:trs}.
  For each participant, column ``\textsf{NO}'s'' reports the
  number of \textsf{NO} answers, \ie of TRSs proved
  non-terminating (with, in parentheses, the number of these
  TRSs detected by the participant only). 
  The \emph{Virtual Best Solver (VBS)} collects the best
  consistent claim for each benchmark since at least 2018;
  column ``\% VBS'' reports the percentage of \textsf{NO}'s
  of each participant \wrt{} the 340 \textsf{NO}'s that the
  VBS collected in total. Then, the following columns report
  the number of \textsf{NO}'s for 4~particular directories
  of~\cite{tpdb}: (E12) corresponds
  to \texttt{TRS\_Standard/EEG\_IJCAR\_12} (49~TRSs),
  (P21) to \texttt{TRS\_Standard/payet\_21} (3~TRSs),
  (P23) to \texttt{TRS\_Standard/Payet\_23} (10~TRSs)
  and (Z15) to \texttt{TRS\_Standard/Zantema\_15} (16~TRSs).
  Globally, the number of non-terminating TRSs detected by
  \nti{} is the second highest, just behind \aprove.
  
  As far as we know, \mnm, \muterm, \natt{} and \tttt{} only
  detect loops. We note that the number of loops found by \mnm{}
  is the highest among these 4~participants.
  It is even very likely that it is the highest among all the
  participants but we could not verify this precisely.
  
  On the other hand, \aprove, \autonon{} and \nti{} are able
  to detect loops as well as other forms of non-termination
  (usually called \emph{non-loops}).
  In order to detect non-loops, \nti{}
  tries to find recurrent pairs while \aprove{} implements
  the approach of~\cite{emmesEG12,oppelt08} and \autonon{}
  that of~\cite{endrullisZ15} (see Section~\ref{sect:rel-work}
  below). To illustrate their technique, the authors of these
  tools submitted several TRSs admitting no loop to the
  organizers of the competition (see directory (E12) for
  \aprove{}, (P21)+(P23) for \nti{} and (Z15) for \autonon). 
  The 10~non-terminating TRSs that \nti{} is the only one
  to detect are those of (P23). Moreover, we note that \nti{}
  finds 3 recurrent pairs in (Z15),
  see Example~\ref{ex:binary-chain-tpdb}, and none in (E12)%
  \footnote{Actually, \nti{} succeeds on 5~TRSs of (E12)
  because it implements a simplistic version
  of~\cite{emmesEG12,oppelt08}},
  while \aprove{} fails on the TRSs of (P21)+(P23) and \autonon{}
  succeeds on those of (P21) only. We guess that \autonon{} fails
  on (P23) because all the TRSs of this directory are not
  left-linear.
  
  The version of \nti{} that participated in the Termination
  Competition~2023 relies on a weaker variation of
  Definition~\ref{def:rec-pair} where $n_2 = 0$
  (actually, the idea of adding $n_2$ to this definition
  came after the competition). Since then, we have extended
  the code to handle $n_2$ and now \nti{} is able to 
  find recurrent pairs for 5~TRSs in the directory
  \texttt{TRS\_Standard/Waldmann\_06} of~\cite{tpdb}
  (\verb+jwno2.xml+, \verb+jwno3.xml+, \verb+jwno5.xml+,
  \verb+jwno7.xml+, \verb+jwno8.xml+) and also for
  the TRS \verb+TRS_Standard/Mixed_TRS/6.xml+. Therefore,
  together with the 3~TRSs in \texttt{TRS\_Standard/Zantema\_15}
  (see Example~\ref{ex:binary-chain-tpdb}), \nti{} is able to
  find recurrent pairs for 9~TRSs that already occurred
  in~\cite{tpdb} before we added (P21) and (P23). All these
  9~TRSs are also proved non-terminating by \autonon{}
  but not by \aprove.
\end{itemize}

\section{Related Work}
\label{sect:rel-work}
In~\cite{payetM06,payet08} we have presented an extended
description of related work in term rewriting and in
logic programming. The state of the art has not evolved
much since then. To the best of our knowledge, the only
significant new results have been introduced
in~\cite{wangS06,oppelt08,endrullisZ15,emmesEG12} for
string and term rewriting.
\begin{itemize}
  \item In~\cite{oppelt08}, the author presents an approach
  to detect infinite chains in string rewriting.
  It uses rules between string patterns of the form $u v^n w$
  where $u$, $v$, $w$ are strings and $n$ can be instantiated
  by any natural number. This idea is extended in~\cite{emmesEG12}
  to term rewriting. String patterns are replaced by
  \emph{pattern terms} that describe sets of terms and have
  the form $t\sigma^n\mu$ where $t$ is a term, $\sigma,\mu$ are
  substitutions and $n$ is any natural number.
  A \emph{pattern rule} $(t_1\sigma_1^n\mu_1, t_2\sigma_2^n\mu_2)$
  is correct if, for all $n\in\nat$,
  $t_1\sigma_1^n\mu_1 \ra^+_P t_2\sigma_2^n\mu_2$ holds,
  where $P$ is the TRS under consideration. Several
  inference rules are introduced to derive correct pattern rules
  from a TRS automatically. A sufficient condition on the derived
  pattern rules is also provided to detect non-termination.
  Currently, we do not have a clear idea of the links between
  the infinite chains considered in these papers and those of
  Section~\ref{sect:binary-chains}. The experimental results
  reported in Section~\ref{sect:implementation} suggest that
  these papers address a form of non-termination which is
  not connected to recurrent pairs.
  Indeed, \aprove{} fails on the TRSs of (P21)+(P23) while
  \nti{} fails to find recurrent pairs for those of (E12).
  \item In~\cite{endrullisZ15}, an automatable approach is
  presented to prove the existence of infinite chains for
  TRSs. The idea is to find a non-empty regular language of
  terms that is closed under rewriting and does not contain
  normal forms.
  It is automated by representing the language by a finite
  tree automaton and expressing these requirements in a SAT
  formula whose satisfiability implies non-termination. 
  A major difference \wrt{} our work is that this technique
  only addresses left-linear TRSs; in contrast, the approach
  of Section~\ref{sect:rec-pairs} applies to any ARS which
  is closed under substitution, and left-linearity is not
  needed. On the other hand, using the notations of 
  Section~\ref{sect:rec-pairs}, the set
  $\{c_1[m,n] \mid m,n \in \nat,\ n \geq n_2\}$
  is a non-empty regular language of terms which is closed
  under $\rra^+_{\Pi}$ and does not contain normal forms
  (see Proposition~\ref{prop:w1-w2}), \ie it is precisely
  the kind of set that the approach of~\cite{endrullisZ15}
  tries to find.
  \item The concept of \emph{inner-looping chain} is
  presented and studied in~\cite{wangS06}. It corresponds
  to infinite chains that have the form
  \[c_1[c_2^{n_0}s\theta^{n_0}] \ra^+_P
  c_1[c_2^{n_1}s\theta^{n_1}] \ra^+_P \cdots\]
  where $P$ is a program, $c_1,c_2$ are contexts, $s$
  is a term, $\theta$ is a substitution
  and the $n_i$'s are non-negative integers (here,
  $c_2^n s \theta^n = c_2[\cdots c_2[c_2[s\theta]\theta]\theta\cdots]$
  where $c_2$ and $\theta$ repeat $n$ times).
  For instance, the infinite chain of 
  Example~\ref{ex:inner-loop} is inner-looping
  because it has the form $a_0 \ra^+_P a_1 \ra^+_P \cdots$
  where, for all $n \in \nat$,
  $a_n = \symb{ab}^n\symb{dc} = c_1[c_2^n s \theta^n]$ 
  with $c_1 = \asym(\square)$, $c_2 = \bsym(\square)$,
  $s = \symb{dc}$ and $\theta = \emptyset$.
  We do not know how inner-looping chains are
  connected to the infinite chains considered in
  Section~\ref{sect:binary-chains}.
\end{itemize}

The notion of \emph{recurrent set}~\cite{guptaHMRX08}
used to prove non-termination of imperative programs is
also related to our work. It can be defined as follows
in our formalism.
\begin{definition}\label{def:rec-set}
  Let $\cA = \ars$ be an ARS.
  A set $B$ is \emph{recurrent for $\cA$} if
  $B \neq \emptyset$, $B \subseteq A$ and
  $(\rra^+_{\Pi}(b) \mathop{\cap} B) \neq \emptyset$
  for all $b\in B$. 
\end{definition}

For instance, in Section~\ref{sect:rec-pairs}, the set
$\{c_1[m,n] \mid m,n \in \nat,\ n \geq n_2\}$ is recurrent
(see Proposition~\ref{prop:w1-w2}). Moreover, the regular
languages of terms computed by the approach
of~\cite{endrullisZ15} are non-empty and closed under
rewriting, hence they are also recurrent.

\section{Conclusion}\label{sect:conclusion}
%
We have considered two forms of non-termination, namely,
loops and binary chains, in an abstract framework that
encompasses term rewriting and logic programming.
We have presented a syntactic criterion to detect
a special case of binary chains and implemented it
successfully in our tool \nti.
As for future work, we plan to investigate the
connections between the infinite chains considered
in~\cite{emmesEG12,wangS06} and in this paper.

\section*{Acknowledgments}
The author is very grateful to the anonymous
reviewers for their insightful and constructive
comments and suggestions on this work.

\bibliographystyle{plain}

\end{document}